\documentclass[final,12pt]{clear2025} 

\usepackage{microtype}
\usepackage{booktabs}
\usepackage{multicol}
\usepackage{subcaption}
\usepackage{caption}
\usepackage{blindtext}
\usepackage{makecell}
\usepackage{xcolor}
\usepackage{dirtytalk}
\usepackage{multirow}
\usepackage{enumitem}
\usepackage{hyperref} 
\usepackage{xcolor}
\usepackage[T1]{fontenc}
\usepackage[utf8]{inputenc}
\usepackage{nicefrac}
\usepackage{enumitem}

\usepackage{glossaries}
\newacronym{rct}{RCT}{randomised controlled trial}
\newacronym{gp}{GP}{Gaussian process}
\newacronym{hte}{HTE}{heterogeneous treatment effect}
\newacronym{ate}{ATE}{Average Treatment Effect}
\newacronym{cate}{CATE}{conditional average treatment effect}
\newacronym{icm}{ICM}{Intrinsic Coregionalization Model}
\newacronym{rwd}{RWD}{Real-World Data}
\newacronym{cf}{CF}{Causal Forest}
\newacronym{rf}{RF}{Random Forest}
\newacronym{rmse}{RMSE}{Root Mean Squared Error}

\title[Causal-ICM]{Causal-ICM: A Data Fusion Framework For Heterogeneous Treatment Effect Estimation With Multi-Task Gaussian Processes}

\clearauthor{%
 \Name{Evangelos Dimitriou} \Email{evangelos.dimitriou.22@ucl.ac.uk}\\
 \addr Department of Statistical Science, University College London
 \AND
 \Name{Edwin Fong} \Email{chefong@hku.hk}\\
 \addr Department of Statistics and Actuarial Science, University of Hong Kong
 \AND
 \Name{Jens Magelund Tarp} \Email{jqmt@novonordisk.com}\\
 \addr Novo Nordisk A/S
 \AND
 \Name{Karla {Diaz-Ordaz}} \Email{karla.diaz-ordaz@ucl.ac.uk}\\
 \addr Department of Statistical Science, University College London
 \AND 
 \Name{Brieuc Lehmann} \Email{b.lehmann@ucl.ac.uk}\\
 \addr Department of Statistical Science, University College London
}

\begin{document}

\maketitle

\begin{abstract}
  Bridging the gap between internal and external validity is crucial for heterogeneous treatment effect estimation. 
  \Glspl{rct}, favoured for their internal validity due to randomisation, often encounter challenges in generalising findings due to strict eligibility criteria. 
  Observational studies, on the other hand, may provide stronger external validity through larger and more representative samples but can suffer from compromised internal validity due to unmeasured confounding. 
  Motivated by these complementary characteristics, we propose a novel Bayesian nonparametric approach, \textit{Causal-ICM}, leveraging multi-task Gaussian processes to integrate data from both \glspl{rct} and observational studies. 
  In particular, we introduce a parameter that controls the degree of borrowing between the datasets and prevents the observational dataset from dominating the estimation. 
  We propose a data-adaptive procedure for choosing the optimal value of the parameter. 
  \textit{Causal-ICM} outperforms other data fusion methods in point estimation across the covariate support of the observational study and provides principled uncertainty quantification for the estimated treatment effects.
  We demonstrate the robust performance of \textit{Causal-ICM} in diverse scenarios through multiple simulation studies and a real-world study.
\end{abstract}

\begin{keywords}
  Data Fusion, Bayesian Nonparametrics, Multitask Gaussian Processes, Generalisability, Heterogeneous Treatment Effects
\end{keywords}

\section{Introduction}

Treatment effect estimation is an important task in many applications, including medicine, epidemiology, and the social sciences. The goal is to quantify the expected impact of a particular intervention in a given target population. The effect of an intervention within a population may vary systematically with respect to a particular set of covariates. Understanding such treatment effect heterogeneity is critical to ensure that appropriate decisions are taken for all individuals, not just those similar to the average \citep{brantner2023methods}. 
An appropriate characterisation of the uncertainty of treatment effect estimates is also required to support robust and reliable decision-making. 
Uncertainty quantification for heterogeneous treatment effects is challenging, however, especially when making inferences about individuals not represented in the study. Classical methods typically rely heavily on extrapolation without adequate variance inflation outside the study support \citep{degitar_review}.
\newline
\hspace*{2em}The gold standard for treatment effect estimation remains the \gls{rct}. The random treatment allocation in the study guards against the effect of confounding 
As a result, under standard identifiability assumptions, we can obtain an unbiased estimate of the treatment effect.
However, \glspl{rct} have important drawbacks \citep{frieden2017}, including high financial costs or the limited sample sizes. While often powered to detect the \gls{ate}, they are typically underpowered for subgroup effects. Furthermore, strict eligibility criteria and selection bias can limit representativeness of the target population \citep{nas2022}. \newline
\hspace*{2em}Observational studies offer an alternative source of data for treatment effect estimation, typically benefiting from larger sample sizes and better representativeness of the target population. However, they may be susceptible to unobserved confounding - unmeasured variables that affect both outcomes and treatment assignment - leading to biased treatment effect estimates. For example, a patient’s underlying health status might influence both the treatment they receive and their eventual outcome, yet remains unmeasured.
\newline
\hspace*{2em}The complementary nature of \gls{rct} and observational data indicates the potential benefits from combining the two sources to obtain better heterogeneous treatment effect estimates. Such data fusion approaches have received increased interest in recent years. A range of methods have been developed to integrate information from multiple sources of data, including a limited number targeting causal inference problems (\cite{colnet_causal_2023}, \cite{lin2024datafusionefficiencygain}). 
These typically rely on strong, untestable assumptions about the structure of the confounding effect and, moreover, are often not able to provide uncertainty quantification around the point estimates of heterogeneous treatment effects.
\newline
\hspace*{2em}In this paper, we introduce a Bayesian nonparametric approach to obtain both point and uncertainty estimates of heterogeneous treatment effects for the target population. Our approach is based on a multi-output \gls{gp}, a natural choice in the data fusion context as it enables joint modelling of the \gls{rct} and observational outcome regression functions, and uses the Bayesian machinery to share information and quantify uncertainty of the functions in regions of differing covariate support. Our key contributions are as follows:
\begin{itemize}[itemsep=-.5ex, topsep=1pt]
    \item We propose \textit{Causal-ICM}, a multi-task \gls{gp} model for heterogeneous treatment effect estimation that accurately captures complex functional relationships in the presence of unobserved confounding both within the support of the \gls{rct} covariate distribution and when extrapolating beyond it.
    
    \item \textit{Causal-ICM} provides principled uncertainty quantification across the full support of the target population. 
    Crucially, we theoretically show that \textit{Causal-ICM} limits the amount of information learnt from the observational dataset, and safeguards against overconfidence in the presence of bias.
    \item Through a comprehensive set of simulation studies and an application to a real-world dataset, we demonstrate that \textit{Causal-ICM} achieves similar or superior performance both in point estimation and uncertainty quantification relative to a broad set of state-of-the-art causal data fusion methods.
\end{itemize}

\section{Related Literature}

\Gls{hte} estimation has received renewed interest in recent years. The most common strategy is to focus on the \gls{cate} function, which quantifies the expected treatment effect given particular covariate values. Under suitable identifiability assumptions \citep{dahabreh_extending_2019}, the \gls{cate} can be estimated parametrically - e.g. using linear regression - or non-parametrically, e.g. using nearest-neighbour matching, kernel methods \citep{brantner2023methods}, or tree-based methods including causal random forests \citep{wager_estimation_2018} and their Bayesian counterpart \citep{hahn_bayesian_2020}. A related strand of work studies semi-parametric and debiased machine learning approaches that combine flexible outcome models with orthogonal score constructions to obtain robust estimates under high-dimensional confounding.

Data fusion across an \gls{rct} and an observational study has tended to focus on \gls{ate} estimation (see \cite{lin2024datafusionefficiencygain} and \cite{colnet_causal_2023} for comprehensive reviews). These methods aim to improve efficiency by integrating information from both sources. \cite{demirel2024predictionpoweredgeneralizationcausalinferences} focus instead on data fusion for generalisability, proposing a framework that builds an observational predictor and then treat the difference between the trial outcome function and that observational predictor as a bias function that captures both confounding and transportation bias due to the differences between the trial and target populations. Similar approaches, although outside the immediate scope of our work, have focused on data fusion for the estimation of long term effects, where randomised data are unconfounded but contain only short term effects, while observational data, although confounded, contain information about long term outcomes (\cite{ghassami2025combiningexperimentalobservationaldata}, \cite{imbens2024longtermcausalinferencepersistent}). 

Data fusion methods for \gls{cate} estimation, the focus of this paper, follow three broad lines of work. One strategy is to use experimental data to debias observational estimates by learning a correction or bias function. \cite{kallus_removing_2018} introduce a two-step approach that corrects hidden confounding in the absence of covariate overlap, relying on the strong assumption of linear confounding and the ability to identify this correction term parametrically. \citet{yang_improved_2022} relax this linearity assumption and provide identifiability and efficiency results under more flexible structural models. \cite{hatt2022combiningobservationalrandomizeddata} develop a representation-learning strategy that learns shared features and confounded outcome models from observational data, using the \gls{rct} to estimate a bias function that debiases these models, with finite-sample bounds that highlight how performance depends on sample size, distribution shift, and bias complexity. \cite{wu_integrative_2022} propose an R-learner that achieves consistency and asymptotic efficiency under covariate overlap between the \gls{rct} and the observational study. A second direction consists of ``test-then-pool'' strategies, in which the trial and observational datasets are combined only if the observational study appears sufficiently unbiased. For instance, \cite{Yang2023Elastic} implement such a strategy, although their approach is limited to linear treatment effect models and relies on strong parametric assumptions. A third line of work uses Bayesian dynamic borrowing to regulate how much information is extracted from the observational study. For instance, \cite{Lin2025} propose a Bayesian dynamic borrowing approach through a power likelihood to reduce the effect of confounding and control the degree of information we borrow from the observational study.

Our method, \textit{Causal-ICM}, takes a Bayesian dynamic borrowing approach based on multi-task \glspl{gp}, whereby each potential outcome is treated as a distinct task. Multi-task \glspl{gp} have been successfully employed for causal inference using solely observational data \citep{alaa_bayesian_2017} and extended to handle hierarchical hidden confounders \citep{witty2020}. Others have built on this framework to develop a general class of counterfactual multi-task deep kernel models that efficiently estimate causal effects and learn policies by stacking coregionalized \glspl{gp} and deep kernels \citep{caron_counterfactual_2022}. \glspl{gp} have also been employed as a matching tool for causal inference \citep{huang2023}. Recent work has also explored Gaussian-process-based partially linear models for heterogeneous treatment effects, combining parametric effect components with flexible GP nuisance estimation \citep{HoriiChikahara2024}. However, these approaches focus on single-source observational settings. In contrast, our work develops a multi-task \gls{gp} framework specifically designed for causal data fusion between experimental and observational sources. Our approach is most similar to concurrent work employing multi-task \glspl{gp} for pseudo-outcome regression to obtain error bounds on the observational bias term \citep{fawkes2025}.
\section{Methodology}

We begin with a brief introduction on data fusion for \gls{hte} estimation. We then introduce a multi-task \glspl{gp} framework tailored for this purpose. Throughout the paper, we use bold-face to denote vectors, e.g. $\mathbf{x}$, and capitals to denote matrices, e.g. $X$.

\subsection{Data Fusion for Heterogeneous Treatment Effect Estimation} \label{sec:fusion}

In this work, we focus on combining information from an observational dataset and a \gls{rct} (experimental) dataset. 
We begin with some notation. Let $S\in\{\text{o},\text{e}\}$ indicate the study, with ‘e’ denoting the experimental sample and ‘o’ the observational one. Let $A\in\{0,1\}$ be the binary treatment indicator. We observe the same baseline covariates in both studies, so each individual has $\mathbf{X}\in\mathcal{X}^S\subseteq\mathbb{R}^p$, where $\mathcal{X}^{\text{e}}$ and $\mathcal{X}^{\text{o}}$ overlap but $\mathcal{X}^{\text{o}}$ is not necessarily contained in $\mathcal{X}^{\text{e}}$. Our methodology requires overlap between the covariate distributions of the \gls{rct} and the observational study, but it does not require the support of the \gls{rct} covariates to be fully contained within that of the observational study. In regions with little or no randomized evidence, the observational study may still provide useful information for generalising treatment-effect estimates, although uncertainty should increase accordingly.

Without loss of generality, we distinguish three possible roles for baseline covariates: treatment-effect modifiers $\mathbf{X}_\tau \in \mathcal{X}_\tau^\text{e}$, confounders $\mathbf{X}_W \in \mathcal{X}_W^\text{e}$, and selection variables $\mathbf{X}_S \in \mathcal{X}_S^\text{e}$, such that $\mathbf{X} = \mathbf{X}_\tau \cup \mathbf{X}_W \cup \mathbf{X}_S$. Here, confounders are variables that act as common causes of treatment and outcome, treatment-effect modifiers are variables indexing the heterogeneity of interest, and selection variables are variables whose distribution may differ between the \gls{rct} and the target population. These roles are not assumed to be mutually exclusive. 
Finally, we observe a continuous outcome $Y\in\mathbb{R}$. The counterfactual outcome $Y_i^a$ denotes the value individual $i$ would have had under treatment $A=a$. 
The causal estimand of interest is the \Gls{cate} in the target population given the baseline covariates ${X}$, defined for $\mathbf{x}_\tau \in \mathcal{X}^{\text{e}} \cup \mathcal{X}^{\text{o}}$ as

\vspace{-1.5em}
\begin{equation*}
    \begin{split}
        \tau(\mathbf{x}_\tau) = \mathbb{E}\left[Y^{a=1} - Y^{a=0}|\mathbf{X}_\tau=\mathbf{x}_\tau\right]
    \end{split}
\end{equation*}
The difference in the conditional expectations of the outcomes in each study is denoted

\vspace{-1.5em}
\begin{align*}
\omega^s(\mathbf{x}_\tau) = \mathbb{E}\left[Y|\mathbf{X}_\tau=\mathbf{x}_\tau, A=1, S=s\right]
-\mathbb{E}\left[Y|\mathbf{X}_\tau=\mathbf{x}_\tau, A=0, S=s\right].
\end{align*}
\vspace{-1.5em}

The fundamental problem of causal inference is that for each unit we cannot observe both potential outcomes simultaneously. Hence, we make the following identifiability assumptions, which are standard in the literature (e.g. \cite{lanners2025datafusionpartialidentification}, \cite{Shi_Pan_Miao_2023}):
\begin{enumerate}[label=A\arabic*, itemsep=0pt, topsep=2pt]
    \item \label{assump:consistency} Consistency: $A=a \Rightarrow Y = Y^a$
    \item \label{assump:A_exchangeability} Mean conditional exchangeability over treatment for the treatment contrast given $\mathbf{X}_\tau$ in the \gls{rct}:
    \begin{equation*}
    \begin{aligned}
    \mathbb{E}[Y^{1}-Y^{0} \mid \mathbf{X}_{\tau}, S=\text{e}]
    &= 
    \mathbb{E}[Y^{1} \mid A=1, \mathbf{X}_{\tau}, S=\text{e}]  - \mathbb{E}[Y^{0} \mid A=0, \mathbf{X}_{\tau}, S=\text{e}] 
    \end{aligned}
\end{equation*}
    \item \label{assump:A_positivity} Positivity of treatment assignment in the \gls{rct}: $P(A=a|\mathbf{X}_S = \mathbf{x}_S, S=\text{e})>0$ for $\mathbf{x}_S \in \mathcal{X}_S^\text{e}$.
    \item \label{assump:S_exchangeability} Mean conditional exchangeability over selection  into the \gls{rct} for the treatment contrast given $\mathbf{X}_\tau$, $\mathbf{X}_S$:
$\mathbb{E}[Y^1 - Y^0 \mid \mathbf{X}_\tau \cup \mathbf{X}_S] = \mathbb{E}[Y^1 - Y^0 \mid \mathbf{X}_\tau \cup \mathbf{X}_S, S=\text{e}] {\color{green}}$
    \item \label{assump:S_positivity} Positivity of \gls{rct} participation: $P(S=\text{e}|\mathbf{X}_S = \mathbf{x}_S)>0$ for $\mathbf{x}_S \in \mathcal{X}^\text{e}_S $
    \item \label{assump: cond_invar_selection} Conditional invariance of selection variable distributions between the \gls{rct} and the target population: $P(X_S|X_\tau) = P(X_S|X_\tau, S=\text{e})$.
\end{enumerate}
Under assumptions \ref{assump:consistency} - \ref{assump: cond_invar_selection}
the \gls{cate} can be identified from the \gls{rct} data as
$\tau(\mathbf{x}_\tau) = \omega^{\text{e}}(\mathbf{x}_\tau)$
for $\mathbf{x}_\tau \in \mathcal{X}_\tau^{\text{e}}$.
The derivation is provided in Appendix~\ref{app:identification}, and a discussion on Assumptions \ref{assump:S_exchangeability} and \ref{assump: cond_invar_selection} is provided on Appendix \ref{app:disc_assump4}. For notational simplicity, and with a slight abuse of notation, we henceforth use $\mathbf{X}$ and $\mathbf{x}$ to denote whichever subset of covariates is relevant in context, including treatment-effect modifiers, confounders, or selection covariates.

Conversely,  relative to the \gls{rct}, the covariate distribution of the observational study may be more representative of the target population. In the presence of unmeasured confounding, however, 
it is not possible to identify $\tau(\mathbf{x})$ from the observational data alone. In particular, the hidden confounding effect in the observational data, denoted by $\eta(\mathbf{x}) = \tau(\mathbf{x}) - \omega^{\text{o}}(\mathbf{x})$, is nonzero. As a result, estimates of $\omega^{\text{o}}(\mathbf{x})$ may be biased for the estimand $\tau(\mathbf{x})$. 

In this work, we use a multi-task Gaussian process to jointly estimate $(\omega^{\text{o}}(\mathbf{x}), \omega^{\text{e}}(\mathbf{x}))$, treating the two data sources as separate tasks. This leverages both the unbiasedness of the \gls{rct} and the broader covariate support of the observational study, while adaptively accounting for confounding by comparing $(\omega^{\text{o}}(\mathbf{x}), \omega^{\text{e}}(\mathbf{x}))$ within $\mathcal{X}^{\text{e}}$. Unlike other methods, we avoid strong parametric assumptions about the bias function, relying instead on the multi-task \gls{gp} to extrapolate non-linearly with posterior uncertainty quantification, so that regions with low covariate support naturally exhibit higher uncertainty.



\subsection{Multi-task Gaussian Processes}
A \gls{gp} is a collection of random variables where any finite subset have a joint Gaussian distribution. In the scalar case, the distribution of a \gls{gp} $f$ is completely specified by its mean function $m\left(\mathbf{x}\right) = \mathbb{E}\left[f\left(\mathbf{x}\right)\right]$ and covariance function $k\left(\mathbf{x}, \mathbf{x}'\right) = \mathbb{E}\left[\left(f\left(\mathbf{x}\right) - m\left(\mathbf{x}\right)\right)\left(f\left(\mathbf{x}'\right) - m\left(\mathbf{x}'\right)\right)\right]$. We will assume that $m(\mathbf{x}) = 0$  throughout. Given a set of observations $(y_i, \mathbf{x}_i)_{i=1}^n$, we seek to learn a function $f$ such that $y_i = f(\mathbf{x}_i) + \epsilon_i$, where we assume $\epsilon_i {\sim} \mathcal{N}(0, \sigma^2)$. \gls{gp} regression proceeds by placing a \gls{gp} prior on the function $f$, and computing the posterior distribution of $f$ given the observations.

In our setting, we use a T-learner approach \citep{K_nzel_2019} to estimate $\tau(\mathbf{x})$, i.e. we individually estimate $\mathbb{E}[Y^0 \mid \mathbf{x}]$ and $\mathbb{E}[Y^1 \mid \mathbf{x}]$ with separate models. For brevity, we will focus on $\mathbb{E}[Y^1 \mid \mathbf{x}]$ in the exposition.
We write the study specific response surfaces as:

\vspace{-1.5em}
\begin{align*}
    f^{\text{e}}\left(\mathbf{x}\right) = \mathbb{E}[Y \mid \mathbf{x},S = \text{e},A = 1], \quad
    f^{\text{o}}\left(\mathbf{x}\right) = \mathbb{E}[Y \mid \mathbf{x},S = \text{o},A = 1].
\end{align*}
\vspace{-1.5em}

For $\mathbf{f}(\mathbf{x}) = (f^{\text{e}}\left(\mathbf{x}\right), f^{\text{o}}\left(\mathbf{x}\right))$, we assign a multi-task \gls{gp} prior 
$\textbf{f} \sim \mathcal{GP}\left(\boldsymbol{\mu}, K\right)$,
where $\boldsymbol{\mu}$ is a bivariate mean function (taken again to be 0) and $K$ is a positive $2 \times 2$ \textit{matrix-valued} covariance function that maps input points $\mathbf{x}_i$ and $\mathbf{x}_j$ to matrices quantifying the relationship of the outputs at these inputs. We specify $K$ in the next section.
The key strength of multi-task \gls{gp}s is the ability to share information between tasks, which we will leverage for data fusion. See Appendix \ref{app:ill_example} Fig \ref{Illustrative example} for an illustrative comparison to independent \gls{gp}s.

We assume a Gaussian observation model,
\begin{align*}
    y_i^{\text{e}}=  f^{\text{e}}\left(\mathbf{x}_i^{\text{e}}\right) + \epsilon_i^{\text{e}}, \quad
       y_j^{\text{o}}=  f^{\text{o}}\left(\mathbf{x}_j^{\text{o}}\right) + \epsilon_j^{\text{o}},
\end{align*}
where $\epsilon_i^{\text{e}}, \epsilon_j^{\text{o}} \sim \mathcal{N}(0,\sigma^2)$ independently across tasks and observations for $i = 1,\ldots,n^{\text{e}}$ and $j= 1,\ldots,n^{\text{o}}$. For simplicity, we assume a common variance $\sigma^2$ between tasks but this can be easily extended to task-specific variances. We write $\mathbf{y}^{\text{e}}$  for the $n^{\text{e}}$-vector of outcomes and $X^{\text{e}}$ for the  $(n^{\text{e}} \times p)$ matrix of covariates respectively, and similarly for the observational dataset.

In data fusion settings for clinical trials, typically $n^{\text{e}}\ll n^{\text{o}}$, and $p(\mathbf{x}^{\text{e}})$ may not have support over the entire target population. Making inferences across the whole population thus requires some level of extrapolation. 
We leverage the multi-task \gls{gp} to improve estimation of $f^{\text{e}}$ outside the support of $p(\mathbf{x}^{\text{e}})$, given the observations $(\mathbf{y}^{\text{o}}, X^{\text{o}})$. 
By sharing information between tasks and leveraging the flexible nature of \glspl{gp}, we hope to increase precision of our estimates without introducing significant bias. 
The object of interest is the posterior distribution of $f^{\text{e}}$ given $\mathcal{D} = (\mathbf{y}^{\text{e}}, X^{\text{e}},\mathbf{y}^{\text{o}}, X^{\text{o}})$, which, by conjugacy, is also a \gls{gp}.  Crucially, the Bayesian framework also provides reliable uncertainty estimates in regions of extrapolation.

\subsection{Intrinsic coregionalisation models (ICMs)}

The \gls{icm} is a special case of a multi-task \gls{gp} with a separable kernel, which is particularly interpretable and suitable for data fusion \citep{icm_vargas_warrick}.  
Each task is a linear combination of independent latent functions, yielding a particularly simple covariance structure. For our setting, we have

\vspace{-1.5em}
\begin{align*}
f^{\text{e}}\left(\mathbf{x}\right) = \sum_{q=1}^Q\alpha_q^{\text{e}} u_q\left(\mathbf{x}\right), \quad f^{\text{o}}\left(\mathbf{x}\right) = \sum_{q=1}^Q\alpha_q^{\text{o}} u_q\left(\mathbf{x}\right),
\end{align*}
where $u_q\left(\mathbf{x}\right)$ are independent zero-mean latent functions taken from the same scalar \gls{gp}, $u_q \sim \mathcal{GP}\left(0, k\right)$ for a chosen kernel $k$. The scalar coefficients $(\alpha_q^{\text{e}},\alpha_q^{\text{o}})$ are  which will be used to construct the coregionalization matrix defined shortly, which governs the dependence structure. The rank, $Q$, of the \gls{icm} is the number of independent components, which we set to $Q = 2$ for the single arm case; we provide a discussion on the interpretation of this later. 

The multi-task covariance function can be expressed as
$$
K(\mathbf{x},\mathbf{x}') = B \otimes k(\mathbf{x},\mathbf{x}'),
$$
where $\otimes$ is the Kronecker product and $B$ is the coregionalization matrix taking values
\begin{align}\label{eq:coreg}
    B &= \begin{bmatrix}
        \beta^{\text{e}} & \beta^{\text{eo}}\\
         \beta^{\text{eo}} &   \beta^{\text{o}}
\end{bmatrix} 
=
    \begin{bmatrix}
        \left(\alpha_1^{\text{e}}\right)^2 +        \left(\alpha_2^\text{e}\right)^2 &
        \alpha_1^{\text{e}}\alpha_1^{\text{o}} + \alpha_2^{\text{e}}\alpha_2^{\text{o}}\\[0.5mm]
\alpha_1^{\text{e}}\alpha_1^{\text{o}} + \alpha_2^{\text{e}}\alpha_2^{\text{o}} &
        \left(\alpha_1^{\text{o}}\right)^2 +        \left(\alpha_2^\text{o}\right)^2
\end{bmatrix}.
\end{align}
The posterior distribution over $f^{\text{e}}$ can then be computed using standard theory (see Appendix \ref{app:posterior}). For a given test point $\mathbf{x}_*$, we have that the posterior distribution is Gaussian:
\begin{align*}
    \left[f^{\text{e}}(\mathbf{x}_*) \mid \mathcal{D}\right] &\sim \mathcal{N}(m^{\text{e}}(\mathbf{x}_*), V^{\text{e}}\left(\mathbf{x}_*\right)),\\
 m^\text{e}(\mathbf{x}_*) = k^{\text{e}}\left(\mathbf{x}_*, X\right)\Sigma^{-1}\mathbf{y}, \quad 
 V^{\text{e}}\left(\mathbf{x}_*\right) &= \beta^{\text{e}}{k}\left(\mathbf{x}_*,\mathbf{x}_*\right) - k^{\text{e}}\left(\mathbf{x}_*, X\right)\Sigma^{-1}k^{\text{e}}\left(X,\mathbf{x}_*\right).
\end{align*}
Here, $\mathbf{y} = (\mathbf{y}^{\text{e}},\mathbf{y}^{\text{o}})$ and $X = (X^{\text{e}},X^{\text{o}})$ are the concatenated response vectors and covariate matrices respectively, and  $\Sigma = K\left(X,X\right)+\sigma^2 I$  with
\begin{align*}
     K(X,X) &= \begin{bmatrix}
       \beta^{\text{e}}K(X^{\text{e}},X^{\text{e}}) &\beta^{\text{eo}}K(X^{\text{e}},X^{\text{o}})\\
       \beta^{\text{eo}}K(X^{\text{o}},X^{\text{e}}) &\beta^{\text{o}}K(X^{\text{o}},X^{\text{o}})
    \end{bmatrix},
\end{align*}
where $K(X^{s}, X^{t})$ is the regular (cross-)covariance matrix of the kernel $k(\mathbf{x},\mathbf{x}')$ for the covariate matrices from $S = s$ and $S = t$.
We similarly have
\begin{align*}
    k^{\text{e}}\left(\mathbf{x}_*, X\right) = \begin{bmatrix}\beta^{\text{e}} k(\mathbf{x}_*,X^{\text{e}})& \beta^{\text{eo}}k(\mathbf{x}_*,X^{\text{o}})\end{bmatrix},
\end{align*}
where $k(\mathbf{x}_*, X^s)$ is the row vector of covariances between $\mathbf{x}_*$ and $X^s$, and $k^{\text{e}}\left( X,\mathbf{x}_*\right) = k^{\text{e}}(\mathbf{x}_*,X)^T$. Given this posterior, $m^{\text{e}}(\mathbf{x}_*)$ provides point estimates of the response surface, with associated posterior variance $V^{\text{e}}(\mathbf{x}_*)$ which can be used to compute credible intervals.


\subsection{\textit{Causal-ICM}}

We now introduce \textit{Causal-ICM}. Specifically, we propose an interpretable parametrisation of the \gls{icm} for causal inference and introduce a bespoke procedure for learning its coefficients, rather than relying on standard marginal likelihood maximisation. This serves to control the influence of the observational dataset and avoid a situation in which a large confounded observational sample yields biased estimates with artificially narrow uncertainty intervals relative to the \gls{rct}. \textit{Causal-ICM} consists of the following choices for the ICM coefficients:
\begin{align*}
    \begin{bmatrix}
    \alpha_1^{\text{e}} &  \alpha_2^{\text{e}}\\
     \alpha_1^{\text{o}} &  \alpha_2^{\text{o}}
    \end{bmatrix} =     \begin{bmatrix}
    1 & 0\\
     \rho &  \sqrt{1-\rho^2}
    \end{bmatrix}
\end{align*}
Under this parameterization, $\rho \in (0,1)$ controls the degree of borrowing, with $\rho \to 1$ corresponding to maximum borrowing and $\rho \to 0$ to none (We provide an additional interpretation of $\rho$ in Appendix \ref{app:interpretation_rho}, Proposition \ref{prop:rho_1}). Although this implies $\beta^{\text{e}} = \beta^{\text{o}} = 1$ (and $\beta_{\text{eo}} = \rho$), the model remains flexible as long as the kernel $k(\mathbf{x},\mathbf{x})$ includes an independent scaling term (e.g., the variance in an RBF kernel) and the scales of $f^{\text{e}}$ and $f^{\text{o}}$ are similar.

With these suggested settings for the ICM coefficients, we then have:
\begin{align*}
f^{\text{e}}(\mathbf{x}) = u_1(\mathbf{x}), \quad f^{\text{o}}(\mathbf{x}) = \rho f^{\text{e}}(\mathbf{x})  + \sqrt{1-\rho^2}u_2(\mathbf{x}).
\end{align*}
Thus, $f^{\text{e}} = u_1$ and $f^{\text{o}}$ is a scaled version of $f^{\text{e}}$ plus a term capturing confounding. Setting $\alpha_2^{\text{e}}=0$ means that the experimental regression surface is modelled directly as a single \gls{gp}, while the observational regression surface decomposes into a shared component and an observational-specific component. Choosing $\rho \in (0,1)$ is reasonable, as $f^{\text{o}}(\mathbf{x})$ is expected to be positively correlated with $f^{\text{e}}(\mathbf{x})$, both model the conditional expectation of the outcome, with $f^{\text{o}}$ additionally contaminated by confounding. In practice, restricting $\alpha_2^{\text{e}}$ minimally affects flexibility while improving interpretability. Remaining hyperparameters, including kernel parameters and sampling variance, are estimated via marginal likelihood maximisation.

\subsubsection{Variance bound for conditional mean function under \textit{Causal-ICM}}

When evaluating treatment effects, it is important to quantify uncertainty alongside point estimates.  
The multi-task Gaussian process provides this directly via the posterior variance $V^{\text{e}}(\mathbf{x}_*)$. In data fusion, a key concern is that large observational samples may lead to biased estimates with unrealistically narrow uncertainty intervals when unmeasured confounding is present. \textit{Causal-ICM} mitigates this, as shown by a lower bound on the posterior variance $V^{\text{e}}(\mathbf{x}_*)$ even as $n^\text{o}$ grows. 
This variance result can be viewed as formally limiting the information flow from the observational dataset, similar to cut models \citep{Lin2025}.

\begin{proposition}\label{prop:v_bound}
    Suppose $\mathbf{x}_*$ is a test point of interest. Let $V^{\emph{e}}(\mathbf{x}_*)$ and ${V}_{\mathcal{D}_{\emph{e}}}^{\emph{e}}(\mathbf{x}_*)$ denote the posterior variances of $f^{\emph{e}}(\mathbf{x}_*)$ given the full dataset $\mathcal{D}=(\mathbf{y}^{\emph{e}}, X^{\emph{e}},\mathbf{y}^{\emph{o}}, X^{\emph{o}})$ and the experimental dataset $\mathcal{D}_{\emph{e}} = (\mathbf{y}^{\emph{e}}, X^{\emph{e}})$ only respectively.
    The posterior variance given $\mathcal{D}$ then satisfies
    \begin{align}
V^{\emph{e}}\left(\mathbf{x}_*\right) \geq (1-\rho^2) \, V^{\emph{e}}_{\mathcal{D}_{\emph{e}}}\left(\mathbf{x}_*\right)
    \end{align}
    where $\rho \in (0,1)$ is the borrowing hyperparameter.
\end{proposition}
\begin{proof}
 Let $n^{\textnormal{e}}, n^{\textnormal{o}} \geq 0$ denote the sizes of the experimental and observational datasets respectively. We outline the proof for $n^{\textrm{e}}= 0$ here,  which corresponds to the prior-only case for the RCT. The proof for the general case $n^{\textrm{e}}> 0$ is deferred to Appendix \ref{app: proof_prop}. The key is that when  $n^{\textrm{e}}= 0$, the posterior variance can be written as
    \begin{align*}
    V^{\textrm{e}}\left(\mathbf{x}_*\right)  = (1-\rho^2)  \, k(\mathbf{x}_*,\mathbf{x}_*) + \rho^2V^{\textnormal{e}}_{\mathcal{D}_{\textnormal{o}}}\left(\mathbf{x}_*\right)
    \end{align*}
    where $V^{\textnormal{e}}_{\mathcal{D}_{\textnormal{o}}}$ is the posterior variance for a \gls{gp} with kernel $k(\mathbf{x},\mathbf{x})$ fit only to $\mathcal{D}_{\textnormal{o}}=(\mathbf{y}^{\textnormal{o}}, X^{\textnormal{o}})$, and $k(\mathbf{x}_*,\mathbf{x}_*)$ is the prior variance. As $V^{\textnormal{e}}_{\mathcal{D}_{\textnormal{o}}}\left(\mathbf{x}_*\right) \geq 0$, we obtain the desired lower bound.
\end{proof}
Proposition~\ref{prop:v_bound} shows that the Causal-ICM posterior variance is lower bounded by a non-degenerate fraction of the experimental-only posterior variance. The same decomposition also yields an upper bound. Writing
\[
V_D^{\mathrm{e}}(x^*) = V_{D_{\mathrm{e}}}^{\mathrm{e}}(x^*) - \rho^2\bigl(k(x^*,X^{\mathrm{o}})-k'(x^*,X^{\mathrm{o}})\bigr)\Sigma_{eo}^{-1}\bigl(k(X^{\mathrm{o}},x^*)-k'(X^{\mathrm{o}},x^*)\bigr),
\]
and noting that the quadratic form is non-negative, we obtain
\[
(1-\rho^2)V_{D_{\mathrm{e}}}^{\mathrm{e}}(x^*) \leq V_D^{\mathrm{e}}(x^*) \leq V_{D_{\mathrm{e}}}^{\mathrm{e}}(x^*).
\]
Hence, incorporating observational data cannot inflate the posterior variance beyond the \gls{rct}-only baseline, nor arbitrarily concentrate the posterior when $\rho\in(0,1)$.

An intuitive understanding of the result is as follows: for two random variables with correlation $\rho$, knowledge of one random variable does not inform us of the exact value of the other unless $\rho = 1$. The key interpretation of Proposition~\ref{prop:v_bound} in the data fusion context is as follows: even as $n^{\text{o}} \to \infty$, the posterior variance of $f^{\text{e}}(\mathbf{x}_*)$ conditional on the full dataset will at most decrease by a factor of $(1-\rho^2)$ relative to the posterior variance given the experimental data only, where $|\rho| \leq 1$.
We conjecture that this inequality is tight as $n^{\text{o}} \to \infty$. This can be interpreted as limiting the `effective sample size' of the observational dataset to be $1/\left(1- \rho^2 \right)$ relative to the experimental dataset. Crucially, this variance bound protects from having very tight credible intervals centred around a biased estimate - the posterior uncertainty will accurately reflect our distrust of the observational dataset.
We also expect this `effective sample size' effect to hold for the posterior mean of $f^{\text{e}}(\mathbf{x}_*)$, where there the bias introduced to the posterior mean by $\mathcal{D}_{\text{o}}$ is limited again by the choice of $\rho$. We leave this interesting direction for future work.  Although our focus is on the conditional mean functions, our approach also provides the posterior distribution of the confounding function, offering additional insights that may be of independent interest (see Appendix \ref{app:confound_fun}).


\subsubsection{Tuning $\rho$}\label{sec:choosing_rho}
We now outline a data-adaptive procedure to select $\rho$, where the goal is to  prevent information from the large confounded observational study from swamping that of the smaller unconfounded RCT. Our proposal is based on
cross-validation to minimize the \gls{rmse} on weighted held-out RCT patients only, where the weighting tailors $\rho$ for extrapolation. Specifically, we use 5-fold cross-validation, as the \gls{rct} sample size is usually small, and we consider $\rho \in [0.0,0.1, \cdots, 1.0]$ on a grid. Our objective is:
\begin{align*}
   \mathcal{L}(\rho)=  \sum_{i = 1}^{n_{\text{held-out}}} w(\tilde{\mathbf{x}}_i)\left(\tilde{y}_i - m^{\text{e}}(\tilde{\mathbf{x}}_i)\right)^2, \quad 
    w(\mathbf{x}) =
 \frac{1}{1-p(S = \text{o} \mid \mathbf{x})}
\end{align*}
where $\{\tilde{y}_i,\tilde{\mathbf{x}}_i\}_{i=1}^{n_{\text{held-out}}}$ is a held-out subset of $(\mathbf{y}^{\textnormal{e}}, X^{\textnormal{e}})$ that was not used to train $m^{\text{e}}$. 
The above objective is motivated as follows. Firstly, we only evaluate predictions on RCT patients where there is no confounding present. Secondly, the weights will upweight \gls{rct} patients who have higher probability of arising from the observational dataset, so  $\rho$ will be tailored to improve extrapolation beyond the \gls{rct} support. We then expect the largest gains in power of $\rho$ is chosen close to 1, indicating that the confounding effect is small. We will see in the experiments that this works well in practice.


\subsubsection{Extension to \gls{cate} Estimation}\label{sec:multi-arm}
When both treatment and control arms are available, we fit two separate rank-2 \glspl{icm} to model $(f_0^{\text{e}}, f_0^{\text{o}})$ and $(f_1^{\text{e}}, f_1^{\text{o}})$, focusing on sharing information between the experimental and observational datasets. This yields independent posteriors for $f_0^{\text{e}}$ and $f_1^{\text{e}}$, which we combine to obtain the posterior of $\tau(\mathbf{x})$ by differencing the means and summing the variances. We allow distinct kernel hyperparameters for the two \glspl{icm} while using a common $\rho$, chosen to minimise the average $\mathcal{L}(\rho)$ across treatment groups. This setup corresponds to a T-learner approach \citep{K_nzel_2019}, and avoids imposing strong assumptions on the bias structure across treatment arms and preserves theoretical tractability. An alternative approach would be to also couple across treatment arms; however, in preliminary experiments not presented here, we found that such strongly coupled models resulted in over-shrinkage of posterior variance and reduced coverage.


\section{Experiments}

We investigated \textit{Causal-ICM}'s performance through multiple simulation studies and a comprehensive analysis of \gls{rwd} sourced from the Tennessee STAR study \citep{star_data}. The code is available online\footnote{https://github.com/EvanDimitriou/CausalICM}.
The optimal value of $\rho$ is chosen data-adaptively in all cases (Section \ref{sec:choosing_rho}). Model evaluation and comparison were based on the \gls{rmse}, averaged over the covariate distribution of the observational study. 

We compared \textit{Causal-ICM} against several benchmarks: (i) a T-learner with \glspl{gp} regressors trained separately on each study, (ii) the two-step method of \cite{kallus_removing_2018}, which debiases observational data via a low-complexity bias function with \glspl{gp} as base learners, (iii) the integrative estimator of \cite{yang_improved_2022}, which assumes a non-linear \gls{cate} and linear confounding effect, (iv) the power likelihood method of \cite{Lin2025}, and (v) the test-then-pool approach of \cite{Yang2023Elastic}. For \cite{yang_improved_2022}, we used the default settings recommended in their paper, while for \cite{Lin2025}, we selected the power parameter $\eta$ by maximizing the Expected Log Pointwise Predictive Density (ELPD), which in both cases resulted in $\eta=0$ (no pooling). Finally, we assessed the coverage of \textit{Causal-ICM}'s \gls{cate} credible intervals across simulation scenarios, comparing them with those from the T-learner approaches.


We used \citet{gpy2014} to train the \glspl{icm}, with $\rho$ tuned as described in Section~\ref{sec:choosing_rho} and remaining hyperparameters estimated via marginal likelihood maximisation. All experiments were conducted on a MacBook Pro (2022) equipped with an Apple M2 chip, 8 CPU cores, 10 GPU cores, and 16 GB of memory, running macOS 26.2. Runtime comparisons indicate that \textit{Causal-ICM} incurs a moderate computational overhead relative to single-study GP baselines, but remains competitive with other data-fusion approaches (see Appendix~\ref{app:runtime}).

\subsection{Simulation Studies}
We showcase the performance of \textit{Causal-ICM} in the main text using two distinct univariate simulation studies. Multivariate simulation studies, are detailed in Appendix \ref{app:multi_simulations}, illustrating the competitive performance of \textit{Causal-ICM} in higher dimensions. In all simulations, the sample size is roughly between 200-300 for the \gls{rct}  and 1000 for the observational study.

In the first simulation setting, we assume a continuous baseline covariate $X\sim \text{Unif}\left[-2,2\right]$. The values of $X$ control the trial participation probability, where $S\sim \text{Bernoulli}\left(p_S\right)$ for $p_S = {\text{exp}\left(-3-3X\right)}/{(1+\text{exp}\left(-3-3X\right)})$. Here, patients with smaller $X$ values have higher probability of being assigned in the trial. Trial participants are randomly assigned treatment with $A \sim \text{Bernoulli}\left(0.5\right)$. The observed outcomes are generated as $Y = A\tau\left(X\right) + X +\epsilon$ for $A\in \{0,1\}$, where $\tau\left(X\right) = 1 + X$ is the \gls{cate} and $\epsilon \sim \mathcal{N}\left(0,1\right)$. For the observational study we have $X \sim \text{Unif}\left[-2,2\right]$, with the treatment generated according to  $A \sim \text{Bernoulli} \left(e\left(X\right)\right)$ where $\text{logit}\left(e\left(X\right)\right) = -X$. The observed outcomes are $Y = A\tau\left(X\right) + X + U +\epsilon$ for $A\in\{0,1\}$. Here, the hidden confounder is generated from $U \sim \mathcal{N}\left(\left(2A-1\right)X, 1\right)$, yielding a linear confounding effect $\eta(X) = 2X$.
In the second simulation, we have non-linear potential outcomes, \gls{cate} and confounding effect, as follows: $Y = A\tau(X) + X^2 - 1 +U + \epsilon$, where $\tau(X) = 1 + X + X^2$ and $U \sim \mathcal{N}((2A-1)\sin(X-1), 1)$, yielding $\eta(X) = 2\text{sin}(X-1)$. All other variables remain the same.

Under a linear \gls{cate} and a linear unobserved confounding effect, \textit{Causal-ICM} with $\rho=0.1$ exhibited comparable performance to the integrative HTE (\cite{yang_improved_2022}) and Elastic HTE estimators (\cite{Yang2023Elastic}), and outperformed all other techniques (Fig~\ref{model_comp} left). In the non-linear setting (both \gls{cate} and unobserved confounding effect), \textit{Causal-ICM} outperformed all other methods yielding the lowest \gls{rmse} (Fig~\ref{model_comp} right). These results reflect the limitations of the experimental grounding, the Integrative HTE and the Elastic HTE methods. Each of these are designed for linear treatment or confounding effect, resulting in inflated \gls{rmse} values in scenarios where these assumptions do not hold. 


Regarding uncertainty quantification, our approach generated 95\% credible intervals with reasonable, though slightly conservative, coverage close to the nominal level in both simulation settings (Fig.~\ref{coverage}). This conservativeness is consistent with the structure of the model: posterior uncertainty reflects not only sampling variability, but also uncertainty induced by partial cross-domain borrowing and the possibility of residual observational bias. The observational-only T-learner had the lowest coverage across the baseline covariate support, which is expected in the presence of unmeasured confounding: as sample size grows, its estimates can remain biased while posterior uncertainty shrinks, leading to overconfident credible intervals that fail to cover the true \gls{cate}. This is exactly the behaviour that Causal-ICM is robust to, as formalized in Proposition~\ref{prop:v_bound}: the large sample size of the observational study does not overwhelm the \textit{Causal-ICM} estimates. See Figure~\ref{extrapolationPlot} in Appendix~\ref{app:extrapolationPlot} for an illustrative example of extrapolation and uncertainty quantification. 
Additional sensitivity analysis regarding imbalance of study sample sizes, performance within and outside the \gls{rct} support, kernel choice and the degree of overlap between the two studies and the choice of $\rho$ is provided in Appendix \ref{app: sensitivity}.

\vspace{-0.5em}
\begin{figure*}[h]
  \centering
  \includegraphics[width=\textwidth]{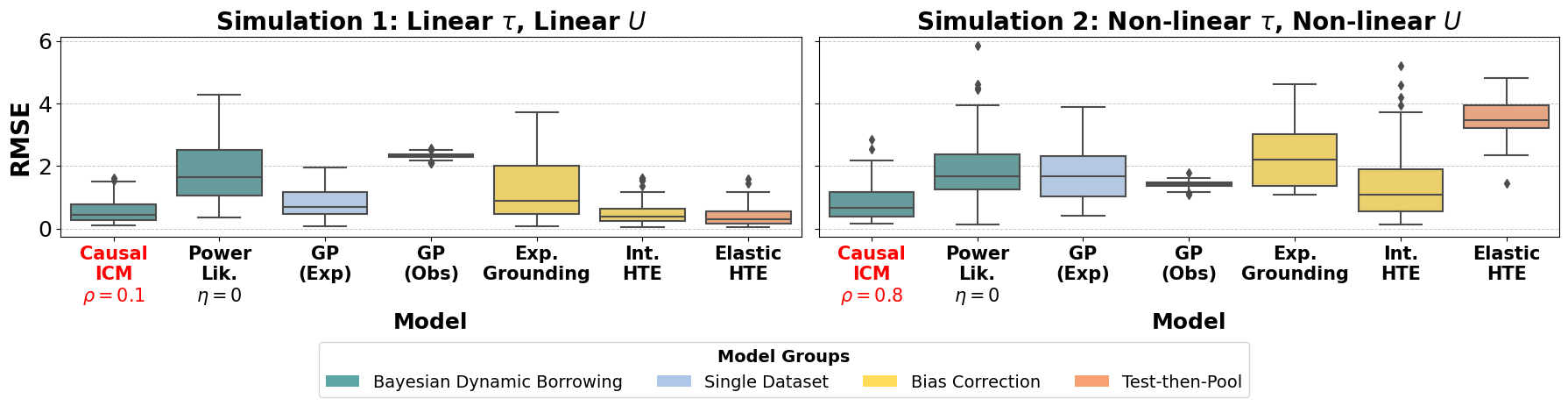}
  \caption{Simulation results over 100 simulated datasets. \textit{Left}: Simulation setting 1: linear \gls{cate} and unobserved confounding, \textit{Right}: Non-linear \gls{cate} and unobserved confounding} 
  \label{model_comp}
\end{figure*}

 \begin{figure}[h!]
  \centering
  \includegraphics[width=\textwidth]{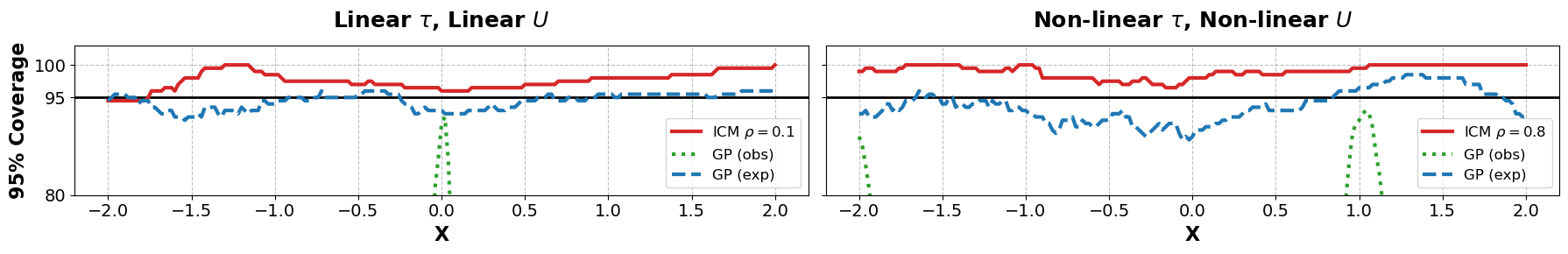}
  \caption{95\% conditional coverage over the covariate distribution of the observational study for the two simulation studies. The black solid line indicates the nominal level (95\%). \textit{Left}: First simulation setting (linear \gls{cate} and confounding), \textit{Right}: Second simulation setting (non-linear \gls{cate} and confounding). The GP (obs) method exhibited 0\% coverage in most of the support covariate distribution.}
  \label{coverage}
\end{figure} 
\vspace{-0.5em}

\subsection{Real-World Data Analysis}
The Tennessee Student/Teacher Achievement Ratio (STAR) \citep{star_data} randomised experiment initiated in 1985 studied the impact of class size on student outcomes through standardized test scores from first to third grade. We focus on two `treatments', here corresponding to two experimental conditions: small class size and regular class size.  We followed a similar approach to \citet{kallus_removing_2018} and used the real data to produce a `confounded' dataset, corresponding to the observational sample, and a smaller unconfounded dataset, which comprised our \gls{rct} data. After removing subjects with missing treatment or outcome values, we were left with a randomised sample of 4218 students. Baseline covariates included gender, race, birth month, birth year, free lunch eligibility and teacher ID. To obtain the `experimental' dataset we sampled randomly only from rural or inner city students. To form the confounded observational dataset, we aggregated all control cases not included in the experimental data. Additionally, for those undergoing treatment, we drew a sample with a down-weighting mechanism applied to individuals whose outcomes fell below the $30^{th}$ percentile. After preprocessing, the unconfounded dataset included 422 students, while the confounded dataset included 2593. Further, 379 students were kept as a validation set. 
We compared \textit{Causal-ICM} (optimal $\rho=0.4$) with a \gls{gp} based T-learner trained either on the experimental or the observational data,the experimental grounding method of \cite{kallus_removing_2018} either with Random Forest or \gls{gp} as base learners and the Integrative HTE method of \cite{yang_improved_2022}. We were unable to include the Elastic HTE estimator (\cite{Yang2023Elastic}) or the power likelihood approach (\cite{Lin2025}), as both methods encountered convergence issues due to the datasets’ high dimensionality and the limited variability in their categorical covariates.

In order to compare \textit{Causal-ICM} with existing methods, we require a notion of ground truth. Given the randomised nature of the study, we assume that an unbiased estimate of the \gls{cate} can be obtained using a doubly robust estimator, which we then treat as the ground truth. To obtain this, we estimated the propensity score and conditional expectation models using the full dataset, prior to splitting it into the confounded and unconfounded samples \citep{saito2020counterfactuala}. We note that, in the absence of the true \gls{cate}, other performance metrics adapted to the counterfactual nature of the evaluation are also available \citep{alaa2019, boyer2023assessing}. The results are summarised in Table \ref{tab:rwd_example}. \textit{Causal-ICM} again performed highly competitively in terms of RMSE, on a par with the experimental grounding method with Random Forests as base learners. 

\begin{table}[h!]
\caption{RMSE values for the RWD example. \textit{Causal-ICM}: our method; GP (exp): T-learner with \glspl{gp} as base learners trained in the unconfounded; GP (obs): T-learner with \glspl{gp} as base learners trained in the confounded data; Experimental grounding (GP): the method proposed by \cite{kallus_removing_2018} with \glspl{gp} as base learners; Experimental grounding (RF): the method proposed by \cite{kallus_removing_2018} with Random Forest as base learners; Integrative HTE: the method proposed by \cite{yang_improved_2022}.}
\label{tab:rwd_example}
\centering
\small
\renewcommand{\arraystretch}{1.0} 
\begin{tabular}{| c | c !{\vrule width 1.5pt} c | c |}
\hline
\textbf{Method} & \textbf{RMSE} & \textbf{Method} & \textbf{RMSE} \\ \Xhline{1.5pt}
\textit{Causal-ICM} ($\rho=0.4$) & \textbf{6.29} & Experimental grounding (GP) & 6.38 \\ \hline
GP (exp) & 6.36 & Experimental grounding (RF) & \textbf{6.27}\\ \hline
GP (obs) & 7.47 & Integrative HTE & 11.84 \\ \hline
\end{tabular}
\end{table}

\section{Discussion}

We have developed \textit{Causal-ICM}, a rank-2 \gls{icm} to combine observational and experimental data to estimate treatment effects for a target population, where the hyperparameter $\rho$ controls the degree of borrowing in an interpretable way. In contrast to existing methods, \textit{Causal-ICM} is highly flexible and does not make assumptions about the functional form of the \gls{cate} or the effect of unobserved confounding. By construction, the method trades a small amount of potential bias from the observational source against gains in precision and external validity, while anchoring inference to the \gls{rct}. This trade-off is further supported by our additional experiments, which show that \textit{Causal-ICM} remains stable under increasing observational sample sizes and continues to perform competitively even when alternative data-fusion approaches deteriorate. By providing accurate estimates of \glspl{hte} together with uncertainty quantification, \textit{Causal-ICM} has the potential to inform more personalised medical decision making. By identifying regions of the population where treatment effects remain highly uncertain, our approach can help guide the design of future trials to improve representativeness and reduce uncertainty to satisfactory levels (\cite{yang2025integratingrctsrwdaiml}).

Despite the inherent difficulty of extrapolation, \textit{Causal-ICM} performs well in these settings, as evidenced by our additional analysis separating performance within and outside the support of the randomized trial, where the gains are particularly pronounced in regions lacking experimental data. \textit{Causal-ICM} also achieves reasonable, if slightly conservative, credible intervals. 
In general, it is challenging to ensure frequentist coverage for Bayesian nonparametric methods.
The slightly conservative empirical coverage of \textit{Causal-ICM} can be explained by the model hedging against disagreement between the experimental and observational signals, with additional uncertainty induced by partial cross-domain borrowing. 
A related limitation of our approach is the sensitivity of uncertainty estimates to the kernel choice and the hyperparameter optimization, itself an intrinsic challenge of \glspl{gp}.

The flexibility of the (multi-task) \gls{gp} framework underlying \textit{Causal-ICM} opens several avenues for future work.  In particular, the current model uses shared kernel hyperparameters across latent functions but more flexible alternatives, such as the Linear Model of Coregionalization (LCM; \cite{alvarez_kernels_2012}), could allow task-specific structure. Other possible extensions include incorporating more than two data sources or treatment arms, and placing a prior on the borrowing parameter $\rho$. More broadly, extending the framework to allow non-Gaussian likelihoods (e.g. for binary or heavy-tailed outcomes) is of significant practical interest, but will likely require non-trivial posterior approximations. 



\bibliography{ref}  

\newpage

\appendix
\section{Identification of \gls{cate} in the RCT} \label{app:identification}

Below we show that the \gls{cate} can be identified using observed data in the \gls{rct}.

\begin{equation*}
\begingroup
\scriptsize
\setlength{\jot}{4pt}
\begin{aligned}
\mathbb{E}[Y^1 - Y^0\mid\mathbf{X}_\tau]
&= \mathbb{E}\!\left[\mathbb{E}[Y^1-Y^0\mid\mathbf{X}_\tau \cup \mathbf{X}_S]\mid\mathbf{X}_\tau\right] \\
&\hfill \text{(tower property)} \\[6pt]
&= \mathbb{E}\!\left[\mathbb{E}[Y^1-Y^0\mid\mathbf{X}_\tau \cup \mathbf{X}_S,S=\mathrm{e}]\mid\mathbf{X}_\tau\right] \\
&\hfill \text{(Assumptions \ref{assump:S_exchangeability})} \\[6pt]
&= \mathbb{E}\!\left[\mathbb{E}[Y^1\mid A=1,\mathbf{X}_\tau \cup \mathbf{X}_S,S=\mathrm{e}]
      - \mathbb{E}[Y^0\mid A=0,\mathbf{X}_\tau \cup \mathbf{X}_S,S=\mathrm{e}]
      \mid \mathbf{X}_\tau\right] \\
&\hfill \text{(Assumptions \ref{assump:A_exchangeability})} \\[6pt]
&= \mathbb{E}\!\left[\mathbb{E}[Y^1\mid A=1,\mathbf{X}_\tau \cup \mathbf{X}_S,S=\mathrm{e}]
 \mid\mathbf{X}_\tau\right]
- \mathbb{E}\!\left[\mathbb{E}[Y^0\mid A=0,\mathbf{X}_\tau \cup \mathbf{X}_S,S=\mathrm{e}]
 \mid\mathbf{X}_\tau\right] \\
&\hfill \text{(linearity of conditional expectation)} \\[6pt]
&= \mathbb{E}\!\left[\mathbb{E}[Y\mid A=1,\mathbf{X}_\tau \cup \mathbf{X}_S,S=\mathrm{e}]
 \mid\mathbf{X}_\tau\right]
- \mathbb{E}\!\left[\mathbb{E}[Y\mid A=0,\mathbf{X}_\tau \cup \mathbf{X}_S,S=\mathrm{e}]
 \mid\mathbf{X}_\tau\right] \\
&\hfill \text{(Assumption \ref{assump:consistency})} \\[6pt]
&= \mathbb{E}[Y\mid A=1,\mathbf{X}_\tau,S=\mathrm{e}]
- \mathbb{E}[Y\mid A=0,\mathbf{X}_\tau,S=\mathrm{e}] \\
&\hfill \text{(tower property and $P(X_S|X_\tau) = P(X_S|X_\tau, S=e)$)}
\end{aligned}
\endgroup
\end{equation*}

\section{Discussion of Assumptions }\label{app:disc_assump4}
Assumption~\ref{assump:S_exchangeability} is weaker than requiring the \gls{rct} to be a random sample from the target population: conditional on $\mathbf{X}_\tau$, it only requires the relevant distribution of selection variables to agree between the trial and the target population. Note that this is strictly weaker than assuming that the \gls{rct} is a random sample from the target population and instead allows for some level of distributional shift. Formally, this means that the \gls{cate} in the target population can be identified as 
\[
\tau(X_\tau) = \mathbb{E}[Y^1 - Y^0|X_\tau] = \mathbb{E}_{X_S|X_\tau}[\mathbb{E}[Y^1 - Y^0|X_\tau, X_S, S=e]]
\]
Thus averaging over $X_S$ given $X_\tau$ in the RCT correctly recovers the \gls{cate}. Importantly, this does not require the marginal distribution of $X_S$ to be identical across studies, and distribution shift on $X_\tau$ (marginally) is allowed.



Although the \gls{cate} can be identified within the \gls{rct} population under assumptions \ref{assump:consistency}-\ref{assump:S_positivity}, this highlights a fundamental limitation of relying on a single data source. The identified estimand, $\tau(X_\tau) = \omega^{(e)}(X_\tau)$, pertains only to the \gls{rct} population and therefore reflects the causal effect among individuals who were eligible for and participated in the trial. In many applied settings, however, the primary scientific or policy objective is to infer treatment effects in a broader target population, which may be better represented by an observational study. Assumption \ref{assump: cond_invar_selection} formalises the conditions under which such transport is possible: within levels of $X_\tau$, the distribution of $X_S$ is invariant between the \gls{rct} and the target population, while still allowing for marginal differences in $X_S$ across the two populations.

\section{Illustrative example}\label{app:ill_example}
Figure \ref{Illustrative example} depicts the differences between single independent \glspl{gp} and multi-task GPs. Independent \glspl{gp} excel in regions with observed data but exhibit suboptimal performance during extrapolation. In contrast, Multitask \glspl{gp} enhance both prediction accuracy and uncertainty quantification in unobserved areas, showcasing superior performance in extrapolation scenarios.

\begin{figure}[h]
 \vskip 0.2in
 \begin{center}
 \centerline{\includegraphics[width=\columnwidth]{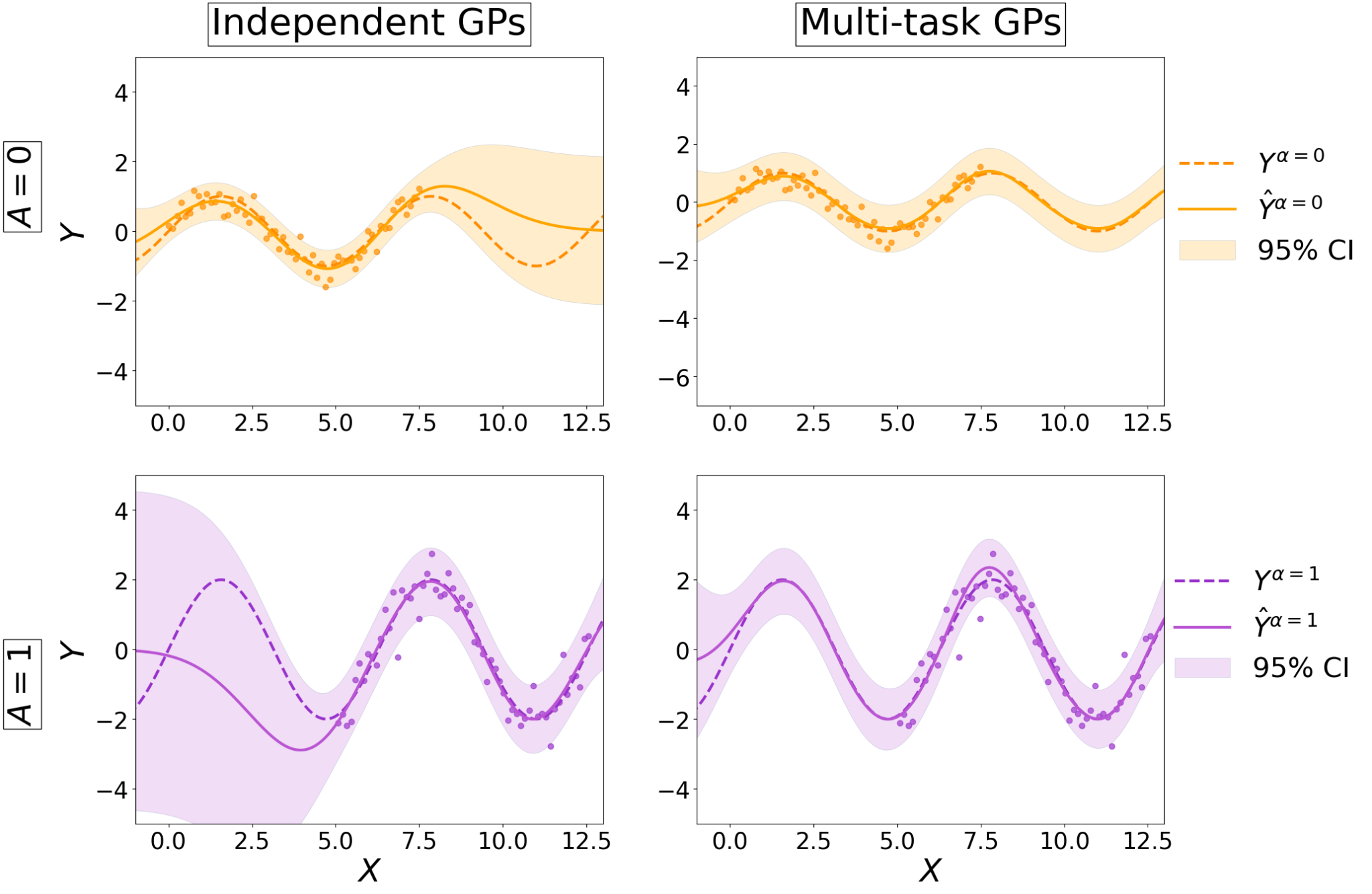}}
 \caption{Comparative illustration between Independent \glspl{gp} and Multitask \gls{gp}. Here, $A$ represents distinct treatment groups, $Y$ corresponds to the outcome, and $X$ denotes baseline covariates.}
 \label{Illustrative example}
 \end{center}
 \vskip -0.2in
 \end{figure}

 \section{Derivation of Posterior Distribution}\label{app:posterior}
The ICM set-up is
\begin{align*}
    f^{\text{e}}\left(\mathbf{x}\right) &= \alpha_1^{\text{e}} u_1(\mathbf{x}) + \alpha_2^{\text{e}}u_2(\mathbf{x})\\
       f^{\text{o}}\left(\mathbf{x}\right) &= \alpha_1^{\text{o}} u_1(\mathbf{x}) + \alpha_2^{\text{o}}u_2(\mathbf{x})
\end{align*}
where $u_i(\mathbf{x}) \sim \mathcal{GP}\left(0, k(\mathbf{x},\mathbf{x})\right)$ independently. Suppose we observe $(\mathbf{y}^{\text{e}}, X^{\text{e}})$ and $(\mathbf{y}^{\text{o}}, X^{\text{o}})$, where the dataset is of size $n_{\text{e}}$ and $n_{\text{o}}$ respectively. The observational model is
\begin{align*}
    y^{\text{e}}_i = f^{\text{e}}(\mathbf{x}_i) + \epsilon_i^{\text{e}}\\
     y^{\text{o}}_i = f^{\text{o}}(\mathbf{x}_i) + \epsilon_i^{\text{o}}
\end{align*}
where both $\epsilon^{\text{e}},\epsilon^{\text{o}}$ arise independently from $\mathcal{N}(0,\sigma^2)$.\\

We clearly have $\mathbb{E}\left[y_i^{\text{e}}\right] = \mathbb{E}\left[y_i^{\text{o}}\right] = 0$, where the expectation is over both the observation noise and the GP prior. The variance is more interesting. It is not difficult to show that
\begin{align*}
    \mathbb{E}\left[y^{\text{e}}_i y^{\text{e}}_j\right] &= \begin{cases}
    \beta^{\text{e}} k(\mathbf{x}^{\text{e}}_i, \mathbf{x}^{\text{e}}_j) + \sigma^2\quad &\text{if } i = j\\
\beta^{\text{e}}k(\mathbf{x}^{\text{e}}_i, \mathbf{x}^{\text{e}}_j)\quad &\text{otherwise}
    \end{cases}
    \\
\end{align*}
Similarly, we have
\begin{align*}
    \mathbb{E}\left[y^{\text{o}}_i y^{\text{o}}_j\right] &= \begin{cases}
\beta^{\text{o}}k(\mathbf{x}^{\text{o}}_i, \mathbf{x}^{\text{o}}_j)+ \sigma^2\quad &\text{if } i = j\\
          \beta^{\text{o}}k(\mathbf{x}^{\text{o}}_i, \mathbf{x}^{\text{o}}_j) \quad &\text{otherwise}
    \end{cases}
\end{align*}
Finally, we have
\begin{align*}
    \mathbb{E}\left[y^{\text{e}}_i y^{\text{o}}_j\right] &= 
    \beta^{\text{eo}}k(\mathbf{x}^{\text{e}}_i, \mathbf{x}^{\text{o}}_j).
\end{align*}
Consider a new test point $\mathbf{x}_*$, and we are interested in $f^{\text{e}}(\mathbf{x}_*)$ and $f^{\text{o}}(\mathbf{x}_*)$. One can also show that $\mathbb{E}\left[f^{\text{e}}(\mathbf{x}_*)\right] = \mathbb{E}\left[f^{\text{o}}(\mathbf{x}_*)\right] =  0$, and 
\begin{align*}
    \mathbb{E}\left[f^{\text{e}}(\mathbf{x}_*) \, y_i^{\text{e}}\right]&=\beta^{\text{e}}k(\mathbf{x}^{\text{e}}_i, \mathbf{x}_*) \\
        \mathbb{E}\left[f^{\text{o}}(\mathbf{x}_*) y_i^{\text{o}}\right]&= \beta^{\text{o}}k(\mathbf{x}^{\text{o}}_i, \mathbf{x}_*)\\
            \mathbb{E}\left[f^{\text{e}}(\mathbf{x}_*) y_i^{\text{o}}\right]&= \beta^{\text{eo}}k(\mathbf{x}^{\text{o}}_i, \mathbf{x}_*)\\ \mathbb{E}\left[f^{\text{o}}(\mathbf{x}_*) y_i^{\text{e}}\right]&=\beta^{\text{eo}}k(\mathbf{x}^{\text{e}}_i, \mathbf{x}_*) 
\end{align*}
We now write this in vector form. Let us define
\begin{align*}
    \mathbf{y} = \begin{bmatrix}\mathbf{y}^{\text{e}}\\\mathbf{y}^{\text{o}}\end{bmatrix}
\end{align*}
which is of length $n_{\text{e}} + n_{\text{o}}$. We can similarly define $X = (X^{\text{e}},X^{\text{o}})$ which is a  $(n_{\text{e}} + n_{\text{o}}) \times p$ matrix.
We then write
\begin{align*}
\mathbf{y}\sim \mathcal{N}\left(\mathbf{0},K(X,X) + \sigma^2 I\right)
\end{align*}
where
\begin{align*}
 K(X,X) = 
    \begin{bmatrix}
\beta^{\text{e}}K(X^{\text{e}},X^{\text{e}}) &\beta^{\text{eo}}K(X^{\text{e}},X^{\text{o}})\\
\beta^{\text{eo}}K(X^{\text{o}},X^{\text{e}}) &\beta^{\text{o}}K(X^{\text{o}},X^{\text{o}})
    \end{bmatrix}
\end{align*}

If we similarly write
$\mathbf{f}(\mathbf{x}_*) = (f^{\text{e}}(\mathbf{x}_*), f^{\text{o}}(\mathbf{x}_*) )$, then we have
\begin{align*}
    \begin{bmatrix}\mathbf{y}\\
    \mathbf{f}(\mathbf{x}_*)\end{bmatrix} \sim \mathcal{N}\left(\mathbf{0},\begin{bmatrix}
        K(X,X)+ \sigma^2 I & K(X,\mathbf{x}_*) \\
        K(\mathbf{x}_*,X) & K(\mathbf{x}_*,\mathbf{x}_*)
    \end{bmatrix}\right)
\end{align*}
where
\begin{align*}
   K(\mathbf{x}_*,\mathbf{x}_*) &= B \otimes k(\mathbf{x}_*,\mathbf{x}_*)\\
   K(X,\mathbf{x}_*) &= \begin{bmatrix}
       \beta^{\text{e}}k(\mathbf{x}_*, X^{\text{e}}) & \beta^{\text{eo}}k(\mathbf{x}_*, X^{\text{e}})\\
      \beta^{\text{eo}}k(\mathbf{x}_*, X^{\text{o}}) &\beta^{\text{o}}k(\mathbf{x}_*, X^{\text{o}}).
   \end{bmatrix}
\end{align*}
To clarify, $ K(X,\mathbf{x}_*)$ is a $(n^{\text{o}}+n^{\text{e}}) \times 2$ matrix.\\

One can then use the usual conditional of a Gaussian distribution, which shows that
\begin{align*}
    \mathbf{f}(\mathbf{x}_*) \mid X, \mathbf{y} \sim \mathcal{N}\left(m_{\mathcal{D}}\left(\mathbf{x}_*\right), K_{\mathcal{D}}\left(\mathbf{x}_*, \mathbf{x}_*\right)\right)
\end{align*}
where
\begin{align*}
    K_{\mathcal{D}}\left(\mathbf{x}_*, \mathbf{x}_*\right) = K\left(\mathbf{x}_*, \mathbf{x}_*\right) -  
    K(\mathbf{x}_*,X)\left[K(X,X) + \sigma^2 I\right]^{-1} K(X,\mathbf{x}_*).
\end{align*}
The posterior mean can similarly be defined as
\begin{align*}\mathbf{m}_{\mathcal{D}}\left(\mathbf{x}_*\right) = K(\mathbf{x}_*,X)\left[K(X,X) + \sigma^2 I\right]^{-1}\mathbf{y}.
\end{align*}

\section{Interpreting $\rho$} \label{app:interpretation_rho}
For further intuition, $\rho$ can be interpreted as a measure of codependence between $f^{\text{e}}$ and $f^{\text{o}}$, as formalised by the following result.
\begin{proposition}\label{prop:rho_1}
    We have $\rho = 1$ if and only if
    \begin{align*}
  \alpha_1^{\emph{e}}\alpha_2^{\emph{o}} =\alpha_1^{\emph{o}}\alpha_2^{\emph{e}}
    \end{align*}
\end{proposition}
\begin{proof}
    It is easier to work with $1-\rho^2$, where plugging in the values from (Equation 1) gives
    \begin{align*}
      1-\rho^2 = \frac{\beta^{\text{e}}\beta^{\text{o}} - \left(\beta^{\text{eo}}\right)^2}{\beta^{\text{e}}\beta^{\text{o}}}  = \frac{\left(\alpha_1^{\text{e}}\alpha_2^{\text{o}} - \alpha_1^{\text{o}}\alpha_2^{\text{e}}\right)^2}{\beta^{\text{e}} \beta^{\text{o}}}
    \end{align*}
   The denominator is positive and finite, so the numerator is thus $0$ if and only if $\alpha_1^{\text{e}}\alpha_2^{\text{o}} =\alpha_1^{\text{o}}\alpha_2^{\text{e}}$.
\end{proof}
Assuming non-zero coefficients, this is equivalent to the condition $\alpha_1^{\text{e}}/\alpha_1^{\text{o}} = \alpha_2^{\text{e}}/\alpha_2^{\text{o}}$, which implies that $f^{\text{e}}$ is a scalar multiple of $f^{\text{o}}$. It is thus intuitive that this results in $\rho = 1$, as learning about $f^{\text{o}}$ is equivalent to learning about $f^{\text{e}}$. Finally, another useful observation is that if a single coefficient is 0, we will have $\rho < 1$ as long as all other coefficients are non-zero.

\section{Proof of Proposition 3.1} \label{app: proof_prop}

Consider first updating the GP with the experimental data points. This gives us
\begin{align*}
    \mathbf{f}(\mathbf{x}_*) \mid X^{\text{e}}, \mathbf{y}^{\text{e}} \sim \mathcal{N}\left(m_{\mathcal{D}_{\textnormal{e}}}\left(\mathbf{x}_*\right), K_{\mathcal{D}_{\textnormal{e}}}\left(\mathbf{x}_*, \mathbf{x}_*\right)\right)
\end{align*}
where
\begin{align*}
    K_{\mathcal{D}_{\textnormal{e}}}\left(\mathbf{x}_*, \mathbf{x}_*\right) &= K\left(\mathbf{x}_*, \mathbf{x}_*\right)-  
   \begin{bmatrix}
       \beta^{\text{e}}k(\mathbf{x}_*, X^{\text{e}})&\beta^{\text{eo}}k(\mathbf{x}_*, X^{\text{e}})
\end{bmatrix}^T\Sigma_{\text{e}}^{-1}\begin{bmatrix}
       \beta^{\text{e}}k(\mathbf{x}_*, X^{\text{e}})&\beta^{\text{eo}}k(\mathbf{x}_*, X^{\text{e}})
   \end{bmatrix}
\end{align*}
where
\begin{align*}
    \Sigma_{\text{e}} = \beta^{\text{e}}K(X^{\text{e}},X^{\text{e}}) + \sigma^2 I
\end{align*}
and  $K(\mathbf{x}_*,\mathbf{x}_*) = B \otimes k(\mathbf{x}_*,\mathbf{x}_*)$.
Note that
$  \begin{bmatrix}
       \beta^{\text{e}}k(\mathbf{x}_*, X^{\text{e}})& \beta^{\text{eo}}k(\mathbf{x}_*, X^{\text{e}})
   \end{bmatrix}$
   is a $n^{\text{e}}\times 2$ matrix. 
The posterior mean can similarly be defined, but is not our focus here. 

We can now simply treat $m_{\mathcal{D}_{\textnormal{e}}}$ and $K_{\mathcal{D}_{\textnormal{e}}}$ as our prior mean and covariance functions respectively, noting that the covariance function can also be written as
\begin{align*}
K_{\mathcal{D}_{\textnormal{e}}}\left(\mathbf{x}_*, \mathbf{x}_*\right) = B \otimes k(\mathbf{x}_*,\mathbf{x}_*) - B' \otimes k'(\mathbf{x}_*,\mathbf{x}_*),
\end{align*}
where 
\begin{align*}
      B' = \begin{bmatrix} 
        (\beta^{\text{e}})^2 & \beta^{\text{e}}\beta^{\text{eo}} \\
        \beta^{\text{e}}\beta^{\text{eo}} & (\beta^{\text{eo}})^2
      \end{bmatrix}, ~~~~~~
      k'(\mathbf{x}_*,\mathbf{x}_*) = k(\mathbf{x}_*, X^{\text{e}})^T {\Sigma_{\text{e}}}^{-1}k(\mathbf{x}_*, X^{\text{e}})
\end{align*}
As in the main paper, let us assume that
\begin{align*}
    \begin{bmatrix}
    \alpha_1^{\text{e}} &  \alpha_2^{\text{e}}\\
     \alpha_1^{\text{o}} &  \alpha_2^{\text{o}}
    \end{bmatrix} =     \begin{bmatrix}
    1 & 0\\
     \rho &  \sqrt{1-\rho^2}
    \end{bmatrix},
    \end{align*}
    which gives
    \begin{align*}
    B = \begin{bmatrix} 
        1 & \rho\\
        \rho & 1
      \end{bmatrix}, \quad     B' = \begin{bmatrix} 
        1 & \rho\\
        \rho & \rho^2
      \end{bmatrix}.
    \end{align*}
If we now observe $(X^{\text{o}},\mathbf{y}^{\text{o}})$, then we can write the full posterior as
\begin{align*}
    \mathbf{f}(\mathbf{x}_*) \mid X^{\text{e}}, \mathbf{y}^{\text{e}},X^{\text{o}}, \mathbf{y}^{\text{o}} \sim \mathcal{N}\left(m_{\mathcal{D}}\left(\mathbf{x}_*\right), K_{\mathcal{D}}\left(\mathbf{x}_*, \mathbf{x}_*\right)\right).
\end{align*}
where the key is that
\begin{align*}
    K_{\mathcal{D}}\left(\mathbf{x}_*, \mathbf{x}_*\right) = K_{\mathcal{D}_{\textnormal{e}}}\left(\mathbf{x}_*, \mathbf{x}_*\right) -  
    {K}^{\text{o}}_{\mathcal{D}_{\textnormal{e}}}(\mathbf{x}_*,X^{\text{o}})\left[{K}^{\text{o}}_{\mathcal{D}_{\textnormal{e}}}(X^{\text{o}},X^{\text{o}}) + \sigma^2 I\right]^{-1} {K}^{\text{o}}_{\mathcal{D}_{\textnormal{e}}}(X^{\text{o}},\mathbf{x}_*)
\end{align*}
where
\begin{align*}
{K}^{\text{o}}_{\mathcal{D}_{\textnormal{e}}}\left(X^{\text{o}},X^{\text{o}}\right) &=
\beta^{\text{o}} K\left(X^{\text{o}},X^{\text{o}}\right) - (\beta^{\text{eo}})^2 K'\left(X^{\text{o}},X^{\text{o}}\right)
\\
&=K\left(X^{\text{o}},X^{\text{o}}\right) - \rho^2 K'\left(X^{\text{o}},X^{\text{o}}\right)
\end{align*}
and
\begin{align*}
{K}^{\text{o}}_{\mathcal{D}_{\textnormal{e}}}\left(X^{\text{o}},\mathbf{x}^*\right)&= \begin{bmatrix}
    \beta^{\text{eo}}k(\mathbf{x}^*,X^{\text{o}})- \beta^{\text{e}}\beta^{\text{eo}}k'(\mathbf{x}^*,X^{\text{o}}),& \beta^{\text{o}}k(\mathbf{x}^*,X^{\text{o}})- (\beta^{\text{eo}})^2k'(\mathbf{x}^*,X^{\text{o}})
\end{bmatrix}
\\
&= \begin{bmatrix}
    \rho\left(k(\mathbf{x}^*,X^{\text{o}})- k'(\mathbf{x}^*,X^{\text{o}})\right),& k(\mathbf{x}^*,X^{\text{o}})- \rho^2k'(\mathbf{x}^*,X^{\text{o}})
\end{bmatrix}
\end{align*}
where ${K}^{\text{o}}_{\mathcal{D}_{\textnormal{e}}}\left(X^{\text{o}},\mathbf{x}^*\right)$ is a matrix of shape $n^{\text{o}} \times 2$ and ${K}^{\text{o}}_{\mathcal{D}_{\textnormal{e}}}\left(\mathbf{x}^*,X^{\text{o}}\right) = {K}^{\text{o}}_{\mathcal{D}_{\textnormal{e}}}\left(X^{\text{o}},\mathbf{x}^*\right)^T$.

We thus have the posterior variance of $f^{\textnormal{e}}$ as
\begin{align*}
    V^{\text{e}}(\mathbf{x}_*) =& \left(k(\mathbf{x}_*,\mathbf{x}_*) -  k'(\mathbf{x}_*,\mathbf{x}_*) \right)\\
    &- \rho^2\left(k(\mathbf{x}^*,X^{\text{o}})- k'(\mathbf{x}^*,X^{\text{o}})\right)^T\Sigma^{-1}_{\textnormal{eo}}\left(k(\mathbf{x}^*,X^{\text{o}})- k'(\mathbf{x}^*,X^{\text{o}})\right)
\end{align*}
where
\begin{align*}
\Sigma_{\textnormal{eo}} &= {K}^{\text{o}}_{\mathcal{D}_{\textnormal{e}}}\left(X^{\text{o}},X^{\text{o}}\right) + \sigma^2 I\\
&=K\left(X^{\text{o}},X^{\text{o}}\right) -\rho^2 K'\left(X^{\text{o}},X^{\text{o}}\right) + \sigma^2 I.
\end{align*}
The first term is the original posterior covariance given $(X^{\text{e}}, \mathbf{y}^{\text{e}})$, whilst the second term is the reduction in variance due to  $\mathcal{D}_{\textnormal{o}} =(X^{\text{o}}, \mathbf{y}^{\text{o}})$, which we want to control. In other words, we want to upper bound
\begin{align*}
\rho^2\left(k(\mathbf{x}^*,X^{\text{o}})-k'(\mathbf{x}^*,X^{\text{o}})\right)^T\Sigma_{\textnormal{eo}}^{-1}\left(k(\mathbf{x}^*,X^{\text{o}})-k'(\mathbf{x}^*,X^{\text{o}})\right)
\end{align*}
We now want to write the above term as the posterior variance of a GP. We can guess the following solution and write:
\begin{align*}
    V^{\text{e}}(\mathbf{x}_*) =& (1-\rho^2)\left(k(\mathbf{x}_*,\mathbf{x}_*) -  k'(\mathbf{x}_*,\mathbf{x}_*) \right)  + \rho^2 P
\end{align*}
where
\begin{align*}
     P&= \left(k(\mathbf{x}_*,\mathbf{x}_*) -  k'(\mathbf{x}_*,\mathbf{x}_*) \right)\\
     &- \left(k(\mathbf{x}^*,X^{\text{o}})-k'(\mathbf{x}^*,X^{\text{o}})\right)^T\Sigma_{\textnormal{eo}}^{-1}\left(k(\mathbf{x}^*,X^{\text{o}})-k'(\mathbf{x}^*,X^{\text{o}})\right).
\end{align*}
If we can show that $P \geq 0$, then we are done. To see this, note that we can apply Woodbury's matrix identity which gives
\begin{align*}
 \Sigma_{\textnormal{eo}}^{-1} &=  [\, \,\underbrace{(1-\rho^2) K'\left(X^{\text{o}},X^{\text{o}} \right)}_{A} +\underbrace{K\left(X^{\text{o}},X^{\text{o}}\right) -K'\left(X^{\text{o}},X^{\text{o}}\right)+ \sigma^2 I}_{B}\,\,]^{-1} \\
 &= B^{-1} - \underbrace{(B + BA^{-1}B)^{-1}}_{C}
\end{align*}
which gives
\begin{align*}
    P&= \left(k(\mathbf{x}_*,\mathbf{x}_*) -  k'(\mathbf{x}_*,\mathbf{x}_*) \right)- \left(k(\mathbf{x}^*,X^{\text{o}})-k'(\mathbf{x}^*,X^{\text{o}})\right)^TB^{-1}\left(k(\mathbf{x}^*,X^{\text{o}})-k'(\mathbf{x}^*,X^{\text{o}})\right)\\
    &+\underbrace{\left(k(\mathbf{x}^*,X^{\text{o}})-k'(\mathbf{x}^*,X^{\text{o}})\right)^TC \left(k(\mathbf{x}^*,X^{\text{o}})-k'(\mathbf{x}^*,X^{\text{o}})\right)}_{\geq 0}
\end{align*}
To show the final term is positive, we just need to show that $C$ is positive semi-definite which implies $x^TC x\geq 0$ for any vectors $x$. If $A$ and $B$ are positive definite (PD) and symmetric, then so is $A^{-1}$ and $BA^{-1}B$. The sum of two PD matrices is PD, as is the inverse of a PD matrix, so $C$ is PD.
Finally, we see that the remaining first terms in $P$ is simply the regular posterior variance (given $\mathcal{D}_{\textnormal{o}}$) of a GP with kernel $k(\mathbf{x}_*,\mathbf{x}_*) -  k'(\mathbf{x}_*,\mathbf{x}_*)$, which is non-negative.

Putting this together then, we have
\begin{align*}
    V^{\text{e}}\left(\mathbf{x}_*\right) \geq (1-\rho^2)  \,{V}^{\text{e}}_{\mathcal{D}_{\textnormal{e}}}\left(\mathbf{x}_*\right)
\end{align*}
where ${V}^{\text{e}}_{\mathcal{D}_{\textnormal{e}}}$ is the variance conditional on $\mathcal{D}_{\textnormal{e}}=(X^{\text{e}},\mathbf{y}^{\text{e}})$ only.



\section{Confounding function}\label{app:confound_fun}
The hidden confounding function $\eta(\mathbf{x}) = f^{\text{e}}(\mathbf{x}) - f^{\text{o}}(\mathbf{x})$ quantifies the degree to which conditional average treatment effect in the observational study deviates from the conditional average treatment effect in the trial. This may be of independent interest in order to understand the underlying process for treatment assignment or allocation in the real-world outside of the experimental setting. Although our principal focus in the above exposition is on $f^{\text{e}}$, we may also easily obtain a posterior distribution for $\eta(\mathbf{x})$, by considering the joint posterior distribution over $\mathbf{f}(\mathbf{x})$.

\subsection{Variance of confounding function}

Suppose we now want to compute the variance of the confounding function evaluated at $\mathbf{x}_*$. This is given by $\eta(\mathbf{x}_*) = f^{\text{e}}(\mathbf{x}_*) - f^o(\mathbf{x}_*)$, or in matrix notation $\eta(\mathbf{x}_*) = \begin{pmatrix}
\phantom{-}1 \\
-1
\end{pmatrix}^T\mathbf{f}(\mathbf{x}_*)$. By standard properties of the multivariate Gaussian distribution, it follows that
\begin{align*}
    \eta(\mathbf{x}_*) \mid X, \mathbf{y} \sim \mathcal{N} \left( m^{\text{e}}(\mathbf{x}_*) - m^o(\mathbf{x}_*), V_\eta(\mathbf{x}_*) \right),
\end{align*}
where
\begin{align*}
    V_\eta(\mathbf{x}_*) &= \begin{pmatrix}
\phantom{-}1 \\
-1
\end{pmatrix}^T K_{*}\left(\mathbf{x}_*, \mathbf{x}_*\right) \begin{pmatrix}
\phantom{-}1 \\
-1
\end{pmatrix} \\
&= (\beta^{\text{e}} + \beta^o - 2\beta^{eo})k(\mathbf{x}_*,\mathbf{x}_*) \\
&- \begin{pmatrix}
(\beta^{\text{e}} - \beta^{eo})k(\mathbf{x}_*, X^{\text{e}}) \\
(\beta^{eo} - \beta^{o})k(\mathbf{x}_*, X^o)
\end{pmatrix}^T \left[K(X,X) + \sigma^2 I\right]^{-1} \begin{pmatrix}
(\beta^{\text{e}} - \beta^{eo})k(\mathbf{x}_*, X^{\text{e}}) \\
(\beta^{eo} - \beta^{o})k(\mathbf{x}_*, X^o)\\
\end{pmatrix}\\
&= 2(1-\rho)k(\mathbf{x}_*, X^{\text{e}})- \begin{pmatrix}
(1-\rho)k(\mathbf{x}_*, X^{\text{e}}) \\
(\rho - 1)k(\mathbf{x}_*, X^o)
\end{pmatrix}^T \left[K(X,X) + \sigma^2 I\right]^{-1} \begin{pmatrix}
(1-\rho)k(\mathbf{x}_*, X^{\text{e}}) \\
(\rho - 1)k(\mathbf{x}_*, X^o)\\
\end{pmatrix}
\end{align*}

\section{Multivariate simulation studies} \label{app:multi_simulations}

Similar to the univariate case, in the multivariate case we vary the complexity of the confounding and the \gls{cate} function, but we choose to explore only cases with nonlinear potential outcomes functions. 

\underline{Simulation 1}: Nonlinear Potential Outcomes, linear confounding, linear \gls{cate}

Assume that we have five continuous baseline covariate $X_j\sim U\left[-2,2\right]$, where $j=1,\cdots, 5$. The values of $X$ define the trial participation mechanism defined as $S\sim \text{Bernoulli}\left(p_S\right)$ where $p_S = \frac{\text{exp}\left(-10-8X_1-8X_2\right)}{1+\text{exp}\left(-10-8X_1-8X_2\right)}$. After selecting all trial participants we assign them randomly to treatment $A \sim \text{Bernoulli}\left(0.5\right)$. The potential outcomes are generated as $Y\left(a\right) = a\tau\left(X\right) + \sum_{j=1}^5X_j +\epsilon$, for $a=0,1$, where $\tau\left(X\right) = 1 + X_1 +X_2 $ is the \gls{cate} and $\epsilon \sim \mathcal{N}\left(0,1\right)$. For the observational study we have $X_j\sim U\left[-2,2\right]$, where $j=1,\cdots, 5$, and the treatment is generated according to the model $A \sim \text{Bernoulli} \left(e\left(X\right)\right)$ where $\text{logit}\left(e\left(X\right)\right) = -\left(X_1+X_2\right)$. Similarly to the trial the potential outcomes are $Y\left(a\right) = a\tau\left(X\right) + \sum_{j=1}^5X_j + U +\epsilon$, for $a=0,1$. $U$ represents the hidden confounding and we generated as $U \sim \mathcal{N}\left(\left(2A-1\right)\left(X_1+X_2\right), 1\right)$. \\

\underline{Simulation 2}: Nonlinear Potential Outcomes, nonlinear confounding, nonlinear \gls{cate}

Assume that we have five continuous baseline covariate $X_j\sim U\left[-2,2\right]$, where $j=1,\cdots, 5$. The values of $X$ define the trial participation mechanism defined as $S\sim \text{Bernoulli}\left(p_S\right)$ where $p_S = \frac{\text{exp}\left(-10-8X_1-8X_2\right)}{1+\text{exp}\left(-10-8X_1-8X_2\right)}$. After selecting all trial participants we assign them randomly to treatment $A \sim \text{Bernoulli}\left(0.5\right)$. The potential outcomes are generated as $Y\left(a\right) = a\tau\left(X\right) + \sum_{j=1}^5X_j +\epsilon$, for $a=0,1$, where $\tau\left(X\right) = 1 + X_1 + X_1^2 +X_2 + X_2^2$ is the \gls{cate} and $\epsilon \sim \mathcal{N}\left(0,1\right)$. For the observational study we have $X_j\sim U\left[-2,2\right]$, where $j=1,\cdots, 5$, and the treatment is generated according to the model $A \sim \text{Bernoulli} \left(e\left(X\right)\right)$ where $\text{logit}\left(e\left(X\right)\right) = -\left(X_1+X_2\right)$. Similarly to the trial the potential outcomes are $Y\left(a\right) = a\tau\left(X\right) + \sum_{j=1}^5X_j + U +\epsilon$, for $a=0,1$. $U$ represents the hidden confounding and we generated as $U \sim \mathcal{N}\left(\left(2A-1\right)\left(\text{sin}\left(X_1\right)+\text{sin}\left(X_2\right)\right), 1\right)$. \\

\begin{figure*}[h]
  \centering
  \includegraphics[width=\textwidth]{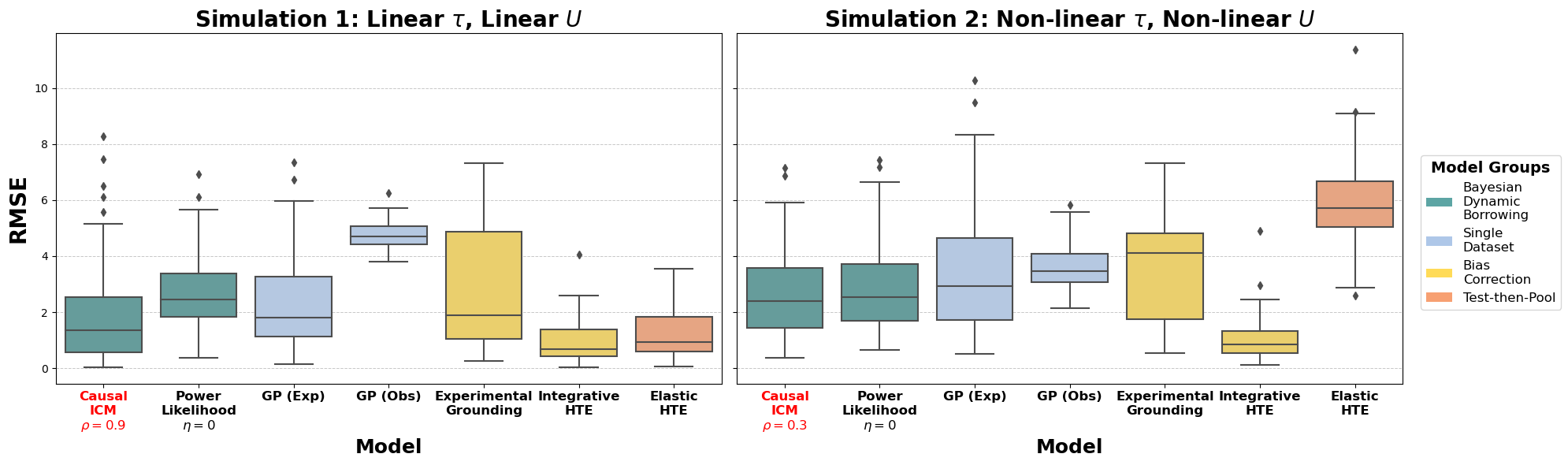}
  \caption{Multivariate simulation results over 100 simulated datasets. \textit{Left}: Simulation setting 1: linear \gls{cate} and unobserved confounding, \textit{Right}: Non-linear \gls{cate} and unobserved confounding} 
  \label{model_comp_mult}
\end{figure*}

Based on Figure \ref{model_comp_mult}, in the first simulation setting, the highest performance is achieved by the Integrative HTE and the Elastic HTE, as expected given that both models are correctly specified. Notably, the performance of \textit{Causal-ICM} is comparable to these benchmarks, demonstrating its competitiveness even when the models are well-specified.

In the second simulation setting, where both the \gls{cate} and the confounding exhibit more complex, non-linear structures, the performance of models relying on correct specification deteriorates. In contrast, \textit{Causal-ICM} maintains robust performance, being outperformed only by the Integrative HTE. Importantly, this superior performance of the Integrative HTE is a consequence of having a correctly specified outcome model; 
In the multivariate simulation settings, parametric competitors can outperform Causal-ICM when their outcome and confounding models are correctly specified. Table~\ref{tab:multivariate-misspecification} shows, however, that under misspecification Causal-ICM becomes more robust and substantially outperforms Integrative HTE in the second multivariate simulation setting.

\begin{table}[h!]
\centering
\caption{RMSE in the second multivariate simulation setting under model misspecification.}
\label{tab:multivariate-misspecification}
\begin{tabular}{lc}
\toprule
Model & Mean RMSE (SD)\\
\midrule
GP (exp) & 3.422 (2.186)\\
GP (obs) & 3.580 (0.796)\\
Experimental Grounding & 3.582 (1.713)\\
Integrative HTE & 21.883 (3.330)\\
Causal-ICM & 2.586 (1.513)\\
\bottomrule
\end{tabular}
\end{table}

\section{Extrapolation with \textit{Causal-ICM}} \label{app:extrapolationPlot}
Below we show the performance of \textit{Causal-ICM} in the first univariate simulation setting, where both the \gls{cate} and the confounding functions are non-linear. The \textit{Causal-ICM} can successfully produce accurate estimates of the \gls{cate} both within and outside the support of the \gls{rct} with inflated uncertainty outside the support.

\begin{figure}[h]
 \vskip 0.2in
 \begin{center}
 \centerline{\includegraphics[width=\columnwidth]{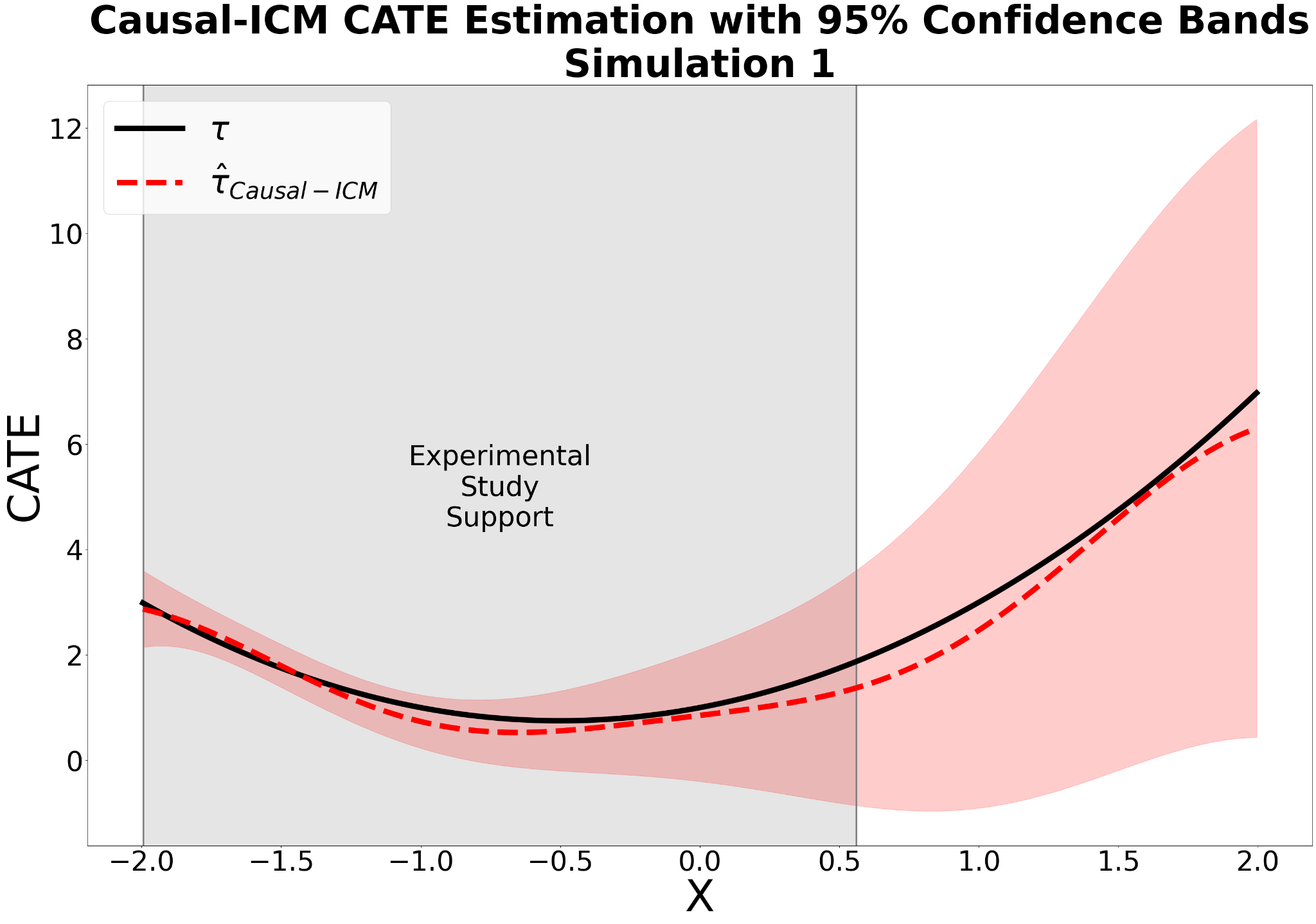}}
 \caption{Estimated CATE under the first univariate simulation setting with nonlinear CATE and confounding functions.}
 \label{extrapolationPlot}
 \end{center}
 \vskip -0.2in
 \end{figure}

\section{Sensitivity to $\rho$, Observational Sample Size, Kernel Choice, and Degree of Overlap; and Performance Within and Outside the RCT Support}
\label{app: sensitivity}

Below we present a more thorough exploration of the effects of different values of $\rho$, differing observational/experimental sample sizes, the choice of kernel and the degree of overlap between the two studies for the second simulation setting of the main text where the \gls{cate} and the confounding function are non-linear.

\begin{table}[h!]
\centering
\begin{tabular}{|c|c|}
\hline
\textbf{$\rho$} & \textbf{RMSE} \\ \hline
0.0 & 2.571 \\ \hline
0.2 & 2.051  \\ \hline
0.4 & 1.502 \\ \hline
0.6 & 0.913 \\ \hline
0.8 & 0.305 \\ \hline
1.0 & 1.095 \\ \hline
\end{tabular}
\caption{RMSE values for different values of $\rho$, Univariate simulation 2}
\label{tab:rmse_rho}
\end{table}

\begin{table}[h!]
\centering
\begin{tabular}{|c|c|}
\hline
\textbf{Kernel} & \textbf{Mean RMSE (SD)} \\ \hline
RBF & 0.822 (0.651) \\ \hline
Matérn 3/2 & 1.315 (0.0556) \\ \hline
Matérn 5/2 & 0.968 (0.583) \\ \hline
\end{tabular}
\caption{Mean RMSE for different kernels, Univariate simulation 2}
\label{tab:rmse_kernel}
\end{table}

\begin{table}[h!]
\centering
\caption{Performance (RMSE) under increasing observational sample size in simulation setting 2. Results are based on 50 simulated datasets.}
\label{tab:sample-imbalance}
\begin{tabular}{lcc}
\toprule
Model & Observational sample size & Mean RMSE (SD)\\
\midrule
Causal-ICM (ours) & 200 & 1.078 (0.810)\\
Integrative HTE & 200 & 1.080 (0.441)\\
Power Likelihood & 200 & 1.023 (0.559)\\
Experimental Grounding & 200 & 1.691 (0.597)\\
\gls{gp} Observational & 200 & 1.512 (0.226)\\
\midrule
Causal-ICM (ours) & 500 & 1.071 (0.652)\\
Integrative HTE & 500 & 1.406 (0.461)\\
Power Likelihood & 500 & 1.231 (0.667)\\
Experimental Grounding & 500 & 1.646 (0.547)\\
\gls{gp} Observational & 500 & 1.475 (0.146)\\
\midrule
Causal-ICM (ours) & 1000 & 0.929 (0.695)\\
Integrative HTE & 1000 & 1.546 (0.303)\\
Power Likelihood & 1000 & 0.975 (0.506)\\
Experimental Grounding & 1000 & 1.652 (0.452)\\
\gls{gp} Observational & 1000 & 1.414 (0.128)\\
\midrule
Causal-ICM (ours) & 2000 & 0.934 (0.674)\\
Integrative HTE & 2000 & 1.576 (0.318)\\
Power Likelihood & 2000 & 1.161 (0.606)\\
Experimental Grounding & 2000 & 1.660 (0.443)\\
\gls{gp} Observational & 2000 & 1.377 (0.062)\\
\bottomrule
\end{tabular}
\end{table}

\begin{table}[h!]
\centering
\begin{tabular}{|c|c|c|}
\hline
\textbf{Overlap} & Optimal $\rho$ &  \textbf{Mean RMSE (SD)} \\ \hline
Full & 0.5 & 0.264 (0.088)\\ \hline
High & 0.4 & 0.444 (0.237) \\ \hline
Low & 0.8 & 0.941 (0.607) \\ \hline
\end{tabular}
\caption{Mean RMSE for different degrees of overlap between the \gls{rct} and observational study for the second simulation setting of the main text (Univariate simulation 2); The degree of overlap is controlled through the coefficients of the model for study participation; Results obtained over 100 different simulated datasets.}
\label{tab:rmse_overlap}
\end{table}

The results above highlight the importance of $\rho$ in controlling the influence of unobserved 
confounding: as $\rho$ increases from 0 to 0.8, the RMSE decreases steadily, reflecting 
a stronger borrowing of information from the observational study, before rising again at 
$\rho = 1.0$ where the confounding is no longer accounted for. Across different kernel 
choices, \textit{Causal-ICM} retains strong and consistent performance, with the RBF kernel 
achieving the lowest mean RMSE, suggesting some robustness to this modelling choice, 
though the differences are moderate.

Table~\ref{tab:sample-imbalance} reports results under increasing observational sample size 
while keeping the experimental sample size fixed. \textit{Causal-ICM} remains competitive 
as the observational sample grows substantially larger than the \gls{rct}, and in fact 
improves with larger observational samples, in contrast to \textit{Integrative HTE} whose 
performance degrades noticeably. The performance of the experimental grounding method of \cite{kallus_removing_2018}, the power likelihood of \cite{Lin2025}, and the T-learner trained on observational data does not seem to deteriorate with the increasing observational sample size, but it remains inferior to our method. Finally, Table~\ref{tab:rmse_overlap} shows that 
\textit{Causal-ICM} performs well even under low overlap between the two studies, 
with optimal results obtained under full overlap, as expected. Taken together, these 
results suggest that \textit{Causal-ICM} is robust across a range of practical conditions 
that may arise when combining experimental and observational data.

To complement the aggregate RMSE results in the main text, Table~\ref{tab:in-out-support} 
reports mean squared error separately inside and outside the support of the \gls{rct} 
for simulation setting~2.

\begin{table}[h!]
\centering
\caption{Mean in-sample and out-of-sample MSE in simulation setting 2. In-sample refers 
to the support of the randomized trial; out-of-sample refers to regions outside the trial 
support.}
\label{tab:in-out-support}
\begin{tabular}{lcc}
\toprule
Method & Mean in-sample MSE (SD) & Mean out-of-sample MSE (SD)\\
\midrule
GP exp & 15.302 (15.595) & 9.645 (13.259)\\
GP obs & 343.387 (436.318) & 310.823 (412.965)\\
Kallus GP & 3.048 (3.141) & 43.461 (42.069)\\
Integrative HTE & 0.307 (0.114) & 3.153 (0.919)\\
Causal-ICM & 0.252 (0.099) & 1.071 (0.722)\\
\bottomrule
\end{tabular}
\end{table}

\textit{Causal-ICM} achieves the lowest in-sample MSE among all methods, and most 
strikingly, retains strong performance outside the \gls{rct} support, where competing 
methods deteriorate considerably. In particular, \textit{Kallus GP} degrades sharply 
out-of-sample despite performing reasonably well within the trial support, suggesting 
it does not generalise well beyond the experimental region. \textit{Causal-ICM}, by 
contrast, benefits from the broader covariate coverage of the observational study to 
extrapolate more reliably, underscoring one of the key practical advantages of 
integrating the two data sources.

\section{Runtime results}
\label{app:runtime}
Training multi-task Gaussian processes scales cubically in the total sample size, although the intrinsic coregionalization structure can often be exploited to improve efficiency in practice. For larger datasets, sparse GP approximations with inducing points provide a natural route to scalability.

\begin{table}[h!]
\centering
\caption{Runtime performance in seconds for simulation setting 2, based on 100 datasets.}
\label{tab:runtime}
\begin{tabular}{lccccc}
\toprule
Method & Mean (s) & SD (s) & Median (s) & Min (s) & Max (s)\\
\midrule
GP exp & 0.244369 & 0.114056 & 0.236458 & 0.089087 & 0.489595\\
GP obs & 2.232287 & 0.639261 & 2.077287 & 1.417782 & 6.207515\\
Experimental Grounding GP & 2.394518 & 0.951331 & 2.181754 & 1.566869 & 10.213725\\
Causal-ICM & 4.559152 & 0.782487 & 4.513439 & 3.310538 & 8.185129\\
Integrative HTE & 6.006100 & 0.416126 & 5.964500 & 5.222000 & 7.494000\\
\bottomrule
\end{tabular}
\end{table}

\end{document}